\documentclass[acmsmall,nonacm]{acmart}

\usepackage[T1]{fontenc}
\usepackage{multicol}
\usepackage{comment}
\usepackage[none]{hyphenat}

\usepackage{graphicx} \usepackage{algorithm}
\usepackage[noend]{algpseudocode}
\usepackage{amsthm,amsmath, mathtools, enumitem, thm-restate, bbm}
\usepackage{color, soul}
\usepackage[dvipsnames]{xcolor}
\usepackage{csquotes}
\definecolor{cobaltblue}{RGB}{0,60,170}
\usepackage{cleveref}
\usepackage{physics}
\usepackage{listings}
\definecolor{codegreen}{rgb}{0,0.6,0}
\definecolor{codegray}{rgb}{0.5,0.5,0.5}
\definecolor{codepurple}{rgb}{0.58,0,0.82}
\definecolor{backcolour}{rgb}{0.95,0.95,0.92}

\lstdefinestyle{mystyle}{
    backgroundcolor=\color{backcolour},   
    commentstyle=\color{codegreen},
    keywordstyle=\color{magenta},
    numberstyle=\tiny\color{codegray},
    stringstyle=\color{codepurple},
    basicstyle=\ttfamily\footnotesize,
    breakatwhitespace=false,         
    breaklines=true,                 
    captionpos=b,                    
    keepspaces=true,                 
    numbers=left,                    
    numbersep=5pt,                  
    showspaces=false,                
    showstringspaces=false,
    showtabs=false,                  
    tabsize=2
}
\crefname{equation}{Eq.}{Eqs.} 
\crefname{section}{Section}{Sections} 
\crefname{figure}{Figure}{Figures} 
\crefname{table}{Table}{Tables} 
\crefname{algorithm}{Alg.}{Algs.}
\crefname{appendix}{Appendix}{Appendices} 
\crefname{theorem}{Theorem}{Theorems}
\crefname{conjecture}{Conjecture}{Conjectures}
\crefname{proposition}{Prop.}{Props.}
\crefname{lemma}{Lemma}{Lemmas}
\crefname{corollary}{Corollary}{Corollaries}

\newcommand{\bool}{\set{0,1}}

\newcommand{\Dc}{\mathcal{D}}
\newcommand{\Ec}{\mathcal{E}}
\newcommand{\Fc}{\mathcal{F}}

\newcommand{\Kc}{\mathcal{K}}
\newcommand{\Sc}{\mathcal{S}}
\newcommand{\Qc}{\mathcal{Q}}

\newcommand{\prob}[2][]{\operatorname*{\mathbb{P}}_{#1 }\brac*{#2}}
\newcommand{\Expect}[2][]{\operatorname*{\mathbb{E}}_{#1 }\brac*{#2}}
\newcommand{\Var}[2][]{\operatorname*{\normalfont{\text{Var}}}_{#1 }\paren*{#2}}
\newcommand{\Cov}[2][]{\operatorname*{\normalfont{\text{Cov}}}_{#1 }\paren*{#2}}

\DeclarePairedDelimiter\ceil{\lceil}{\rceil}
\DeclarePairedDelimiter\brac{\lbrack}{\rbrack}
\DeclarePairedDelimiter\set{\lbrace}{\rbrace}
\DeclarePairedDelimiter\paren{\lparen}{\rparen}

\newtheorem{theorem}{Theorem}

\newtheorem{lemma}[theorem]{Lemma}
\newtheorem{corollary}[theorem]{Corollary}
\newtheorem{conjecture}[theorem]{Conjecture}

\theoremstyle{definition}

\theoremstyle{remark}

\newcommand{\nth}{^\text{th}}

\newcommand{\degb}[1]{d_{#1}^{\rightarrow}}
\newcommand{\tgk}{t^{>k}}
\newcommand{\tlk}{t^{<k}}
\newcommand{\Tgk}{T^{>k}}
\newcommand{\Tlk}{T^{<k}}
\newcommand{\Rgk}{R^{>k}}
\newcommand{\Rlk}{R^{<k}}

\newcommand{\create}{\mathtt{create}}
\newcommand{\Insert}{\mathtt{insert}}
\newcommand{\queryedge}{\mathtt{query\_edges}}
\newcommand{\QS}{\mathcal{Q}_\mathcal{S}}

\newboolean{showcomments}
\setboolean{showcomments}{true}
\newcommand{\HL}[1]{\ifthenelse{\boolean{showcomments}}{{\color{SeaGreen}{HL: #1}}}{}}
\newcommand{\SK}[1]{\ifthenelse{\boolean{showcomments}}{{\color{blue}{SK: #1}}}{}}

\newcommand{\ER}{Erd\H{o}s--R\'enyi\ }
\newcommand{\Delone}{\Delta_{E_1}}

\title{A Perfectly Distributable Quantum-Classical Algorithm for Estimating Triangular Balance in a Signed Edge Stream}
\date{February 16, 2026}

\makeatletter
\let\@authorsaddresses\@empty
\makeatother

\author{Steven Kordonowy}
\affiliation{\institution{Fujitsu Research of America Inc., Santa Clara}
            \country{USA} }
\affiliation{\institution{University of California at Santa Cruz}
            \country{USA} }
\author{Bibhas Adhikari}
\author{Hannes Leipold}
\affiliation{\institution{Fujitsu Research of America Inc., Santa Clara}
             \country{USA}}

\begin{document}
\raggedbottom 

\begin{abstract}

We develop a perfectly distributable quantum-classical streaming algorithm that processes signed edges to efficiently estimate the counts of triangles of diverse signed configurations in the single pass edge stream. 
Our approach introduces a quantum sketch register for processing the signed edge stream, together with measurement operators for query-pair calls in the quantum estimator, while a complementary classical estimator accounts for triangles not captured by the quantum procedure. 
This hybrid design yields a polynomial space advantage over purely classical approaches, extending known results from unsigned edge stream data to the signed setting.
We quantify the lack of balance on random signed graph instances, showcasing how the classical and hybrid algorithms estimate balance in practice. \end{abstract}

\maketitle


\section{Introduction}\label{sec:intro}

Computing statistical quantities for summarizing large-scale data, particularly large networks, under memory-constrained resources has become increasingly important in the era of big data. Quantum computing systems could mature to play a pivotal role in this evolving landscape~\cite{kallaugher_how_2024}. For example, in the single pass streaming setting, quantum devices have shown to provide space advantages in a host of settings, including one-way Boolean functions~\cite{gavinsky_exponential_2006}, triangle counting~\cite{kallaugher_quantum_2022-1}, and directed maxcut~\cite{quantum_maxdicut_2023}. Despite the abundance of distributed computing frameworks for classical streaming algorithms, the literature lacks quantum distributed frameworks for quantum streaming algorithms. In this paper, we attempt to fill this gap by developing a quantum-classical streaming algorithm for counting triangles in a signed edge stream implemented in a distributed computing setup.

Triangle finding and counting in the edge streaming graph model is a fundamental problem due to its importance in network analysis and algorithms. In this model, edges of a graph arrive sequentially as a stream, and algorithms are required to detect or estimate the number of triangles using sublinear space, often in a single pass. Early work established both lower bounds and randomized algorithms for triangle counting in insertion-only streams, highlighting the inherent space accuracy trade-offs~\cite{bar2002counting}. Subsequent approaches introduced sampling-based and sketching techniques to approximate triangle counts in massive graphs~\cite{buriol2006counting, jha2015streaming}. The algorithm described in Ref.~\cite{jayaram2021optimal}, henceforth the Jayaram-Kallaugher (JK) algorithm, is provably optimal in space (up to polylogarithmic factors), settling the problem on classical triangle counting in the edge stream setup by matching known lower bounds under standard parameterizations of the problem. 
In Ref.~\cite{kallaugher_quantum_2022-1}, Kallaugher introduced a quantum-classical algorithm for approximating triangle counts that results in a polynomial space advantage over the classically-optimal JK algorithm.

In this paper, we consider the single-pass insertion-only signed graph stream model in which each signed edge appears only once. Signed graphs have two types of edges: positive and negative edges (see \cref{Sec:2} for further details). Applications of signed graphs and balance measures are widespread. 
In social network analysis, they model trust/distrust relations and help identify polarized communities~\cite{leskovec2010signed,facchetti2011computing}. 
In international relations, signed networks capture alliances and conflicts among countries, with balance dynamics linked to stability~\cite{antal2005dynamics,marvel2009energy}. 
In finance, signed correlation graphs describe co-movements and anti-movements of asset returns, with balance analysis providing insights into systemic risk and market stability~\cite{boginski2005statistical,kenett2010dominating,bargigli2015multiplex,bartesaghi2025global,adhikari2025signed, adhikari2026signed}. 
In biology, signed networks model cooperative and competitive interactions in gene regulatory and protein networks~\cite{estrada2010balance,saade2014spectral}. 
In physics and complex systems, balance theory underpins models of consensus and cooperation under antagonistic interactions~\cite{altafini2013consensus,traag2013dynamical}.  

\begin{figure}[!t]
\includegraphics[width=1.0\textwidth]{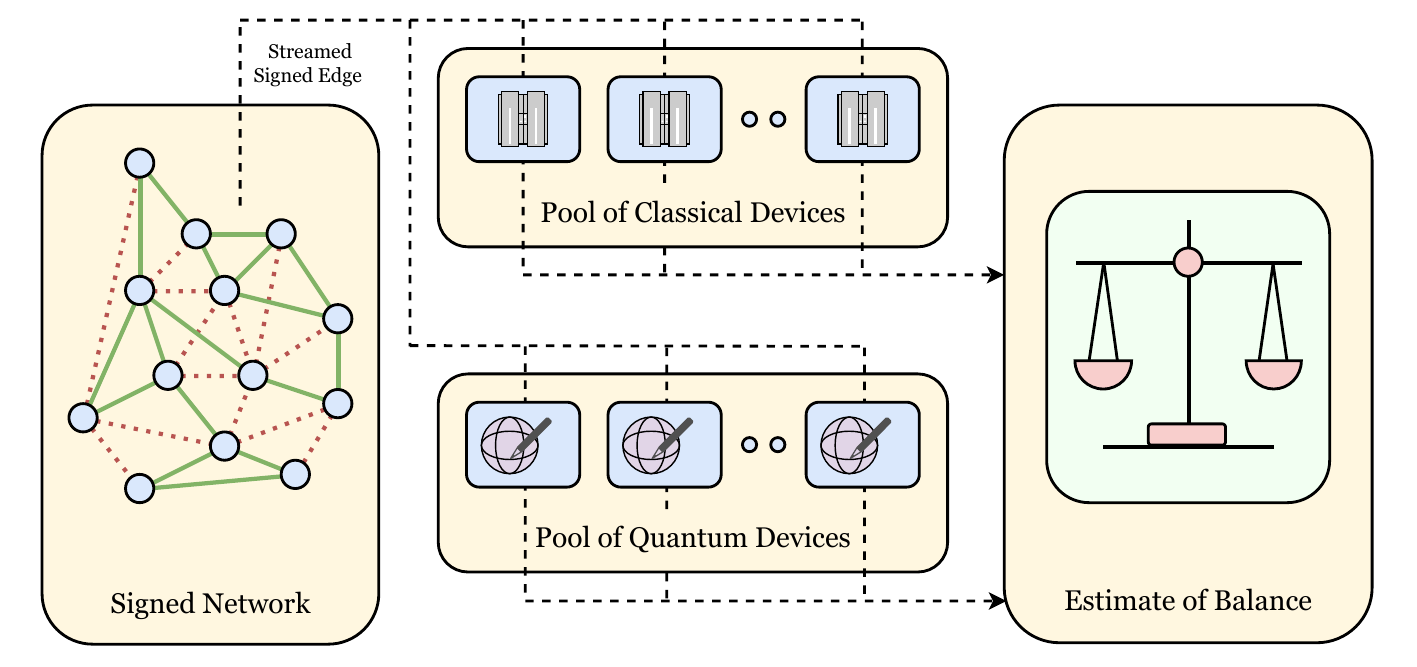}
\caption{\textbf{Quantum-Classical Estimation of Balance in the Streaming Model.} An overview of our hybrid quantum-classical algorithm for estimating balance. 
}
\label{fig:flow_signed}
\end{figure}

In the classical setting, we extend the classically-optimal streaming JK algorithm for triangle counting in edge streams to the setting of signed edge streams and we call it \textit{signed} JK \textit{algorithm}. In this model, each edge in the stream is associated with a sign, and consequently, triangles can be classified into four distinct types (see \cref{fig:sign_triangle}). This introduces additional challenges, as the algorithm must track the product of edge signs in order to correctly identify the type of each triangle during query processing. The resulting estimates of the different triangle types are then used to approximate the degree of balance of the underlying signed graph, as defined by the triangular measure of balance.

Further, we extend the area of space advantage from quantum sketchpads for important and realistic data streaming setting by developing an algorithm for counting triangles for \textit{signed networks}, focusing specifically on the statistical characterization of \textit{balance}.  We then present a hybrid quantum-classical algorithm that gives a quantum polynomial advantage over the signed JK algorithm. 
The proposed algorithm relies on two complementary estimators: a quantum estimator and a classical estimator. The quantum estimator employs a collection of quantum processors, each performing analogous tasks involving sampling and the execution of quantum queries, while the classical estimator consists of a collection of classical processing units carrying out a similar sampling-based procedure. These two estimators play complementary roles within the hybrid framework. In particular, triangles of a given type that may be missed by the quantum estimator due to the probabilistic nature of quantum measurements are detected by the classical estimator with nonzero probability.
The overall hybrid approach is inspired by the classical–quantum algorithm for estimating triangle counts in unsigned edge streams proposed in Ref.~\cite{kallaugher_quantum_2022-1}. The primary technical challenge addressed in this work lies in the construction of a quantum sketch register capable of encoding signed edges observed in the signed edge stream, together with the design of appropriate POVM operators for performing quantum queries on this register. A circuit-level description of the proposed procedure is provided in \cref{fig:q_sketch_circ}. We then use this approximation of triangle types to estimate the balance of networks. 

Finally, we introduce a distributed computing framework that integrates classical and quantum processors to implement the proposed hybrid quantum-classical algorithm for estimating different types of signed triangles which is leveraged to approximate the structural balance of the underlying signed graph. To evaluate performance, we generate signed edge streams from random signed graphs and compare the proposed hybrid approach with its purely classical counterpart using state-vector simulations. To the best of our knowledge, this work represents the first instance in the literature of a distributed quantum-classical computing model designed specifically for implementing streaming algorithms on graph streams. The framework developed in the paper is presented in Figure \ref{fig:flow_signed}.

We emphasize that the proposed distributed framework is fundamentally different from standard distributed  frameworks. Classical approaches rely on interconnected machines and must address challenges such as \textit{parallelism}, \textit{load balancing}, and \textit{network efficiency}~\cite{meng2024survey}. Standard distributed quantum computing typically focuses on quantum networking challenges, including implementing quantum operations such as quantum gates across multiple quantum processing units (QPUs). A \emph{perfectly distributable} quantum algorithm is one in which each processor contributes independently without inter-processor communication~\cite{caleffi2024distributed}. We establish that both the hybrid algorithm of Ref.~\cite{kallaugher_quantum_2022-1} and the algorithm proposed here fall into this category: a pool of QPUs and classical processing units (CPUs) operates independently, and the final estimate is obtained via combined classical post-processing of outcomes (see Section~\ref{Sec:5}).

\section{Graph streams, triangles and balance of Signed Graphs in a Stream}\label{Sec:2}

In this section, we briefly review of existing distributed computing frameworks for graph streams and structural balance theory of signed graphs. We also include a discussion on the set of parameters used in sequel. 

\begin{figure}[!t]
\centering
\includegraphics[width=0.6\textwidth]{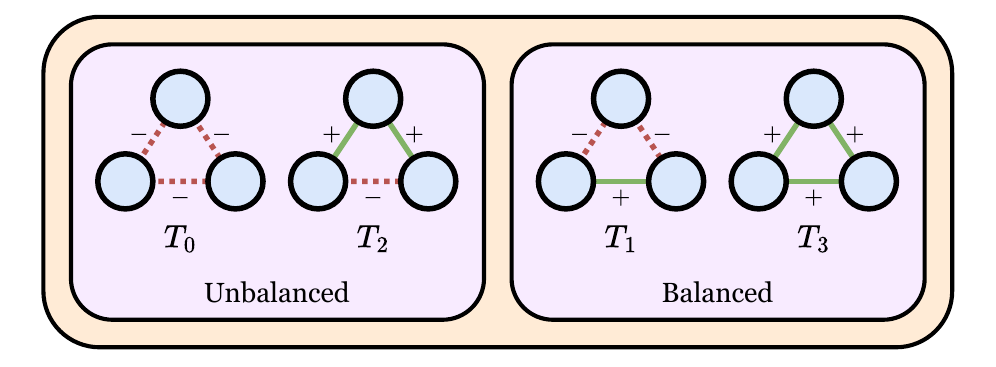}
\caption{\textbf{Triangle Types and their Balance in a Signed Network.} Triangle types labeled by their number of positive edges, with triangles with odd number of plus edges categorized as balanced while those with even positive edges are categorized as unbalanced. Under structural balance theory, real-life network dynamics may evolve towards more balanced local substructures. As such, counting triangle types is essential for analyzing networks.}
\label{fig:sign_triangle}
\end{figure}

\subsection{Graph Streams}

In classical streaming data analysis model, the \textit{graph streaming model} extends the classical data stream paradigm to graph-structured data, where vertices and edges arrive sequentially as a stream. In this setting, the input graph is too large to store explicitly in memory, and algorithms are required to process updates, such as edge insertions, deletions, or weight changes, using sublinear space and typically a single or a small number of passes over the stream. 
Depending on the update rules, graph streams are commonly classified into insertion-only, turnstile, and dynamic models. The graph streaming framework has become a central tool for analyzing massive networks arising in social media, communication systems, and biological data,  establishing space-accuracy tradeoffs for fundamental graph problems~\cite{ahn2012graph, mcgregor2014graph, kapralov2013streaming, cormode2010counting, woodruff2014sketching}.

Classical streaming algorithms, in particular, graph streaming algorithms are deployed naturally in a distributed computing framework for massive graph analytics, with theoretical guarantees on space, communication complexity, and approximation accuracy~\cite{muthukrishnan2005streams,ahn2012graph, mcgregor2014graph,meng2024survey}. Several open-source distributed graph processing frameworks such as the vertex-centric programming models: Google's Pregel~\cite{malewicz2010pregel}, Pregel+~\cite{yan2015effective}, GraphLab~\cite{low2012distributed}, PowerGraph~\cite{gonzalez2012powergraph}, Giraph~\cite{han2015giraph}, and GPS~\cite{salihoglu2013gps} have been developed considering the abstraction of parallel looping, message receiving and sending, and bradcasting. Besides, edge-centric and subgraph-centric distributed graph processing frameworks are proposed to boost performance that require edge lists and bypassing random memory accesses associated with vertices, such as X-Stream~\cite{roy2013x}, WolfGraph~\cite{zhu2020wolfgraph}, Giraph++~\cite{tian2013think}, GoFFish~\cite{simmhan2014goffish}, Blogel~\cite{yan2014blogel}.

\subsection{Signed Graphs and Signed Triangles}

A signed graph $G = (V,E,\sigma)$ is a graph where each edge $(i,j) \in E$ is assigned a sign $\sigma(i,j) \in \{+1,-1\}$, representing positive (friendly, cooperative, positively correlated) or negative (hostile, antagonistic, negatively correlated) relationships~\cite{harary1953notion,zaslavsky2012mathematical}. 
A signed network can be analyzed by counting the number of small substructures contained in it; the simplest example being the number of triangles. Since the network is signed, the type of triangles are characterized by the sign of the three edges in the triangle, as depicted in \cref{fig:sign_triangle}. There are four types of triangles in a signed graph, depending on the number of \textit{positive} edges in the triangle. Define the enumeration of each type of triangle as 
\begin{align} \label{eq:tri_types}
T_{j}, \textit{the number of triangles with $j$ positive edges in $G$}. 
\end{align}

\noindent Then the total number of triangles is the sum of each type:
\begin{align}\label{eq:tri_total}
T = T_{0} + T_{1} + T_{2} + T_{3}, \textit{the number of triangles in $G$}.
\end{align}

\subsection{Structural Balance Theory}

The classical notion of balance, introduced by Harary, states that a signed graph is balanced if its vertex set can be partitioned into two disjoint subsets such that all positive edges lie within the subsets and all negative edges lie between them. 
Equivalently, a signed graph is balanced if every cycle has a positive sign, i.e., the product of the edge signs along the cycle is $+1$.  

Several quantitative measures of balance have been proposed. 
One widely used measure is the \emph{balance index}, defined as 
\begin{equation}\label{eqn:defbal}
\mathcal{B}(G) = \frac{\#\{\text{balanced cycles in } G\}}{\#\{\text{total cycles in } G\}},    
\end{equation}
which captures the proportion of balanced structures in the network~\cite{cartwright1956structural}. 
Another is the \emph{frustration index}, which is the minimum number of edge sign changes (or deletions) required to make the graph balanced \cite{brusco2019relaxations}.  

Then, in structural balance theory, triangles of types $T_1$ and $T_3$ are balance, whereas triangles of type $T_0$ and $T_2$ are unbalance. Thus, a balance signed graph has triangles of types $T_0$ and $T_2$ only, and hence lack of balance of a signed graph can be estimated by counting all types of triangles in the graph. 
In particular, the \textit{triangular index of balance} of a signed graph $G$ defined as
\begin{equation}\label{eqn:defbalT}
\mathcal{B}_{\triangle}(G) = \frac{\#\{\text{balanced triangles in } G \}}{\#\{\text{total triangles in } G \}}.   
\end{equation} 

Since a triangle is a cycle with three nodes, $\mathcal{B}_{\triangle}$ is a special type of $\mathcal{B}$ (see \cref{eqn:defbal}).

\subsection{Parameter Approximation in the Stream}

A memory limited device does not have access to the entire input data when solving a problem, since the data might be too large to fit into memory at once, such as in very large social media networks, or due to the nature of the data is that it comes in periodically throughout a period of time, such as in trades in a financial system. 
In the single-pass streaming setting, we are given a stream $S = ((uv, \sigma_{uv})_i)_{i=1}^m$ composed of edges $uv \in E(G)$ and their sign $ \sigma_{uv} \in \{+1,-1\} $.
We wish to (approximately) count $T_{j}$ (from \cref{eq:tri_types}) in the network $G$. 
We work in the insertion-only model in which edges are processed one at a time and are never deleted. 
Given error parameters $\epsilon, \delta \in (0,1]$, the task is to find a $(1 \pm \epsilon)$-multiplicative approximation of $T$ with probability $1-\delta$ using as few space resources as possible. 

Any algorithmic approximation guarantee should be independent of the ordering of $S$. Indeed, to show worst-case guarantees, one may construct adversarial orderings of $S$. In order to construct hard graph instances, we define these standard graph parameters for a graph $G$:
\begin{multicols}{2}
\begin{itemize}[]
\item[] $n := |V|$, \textit{number of nodes in the graph}
\item[] $\Delta_V$, \textit{maximum number of triangles incident on any node}
\item[] $m := |E|$, \textit{number of edges in the graph}
\item[] $\Delta_E$,  \textit{maximum number of triangles incident on any edge}
\end{itemize}
\end{multicols}
\noindent The space complexity for state-of-the-art streaming algorithms for triangle counting are given with respect to these parameters. Note that $T, \Delta_E, \Delta_V$ are not known ahead of time, since that would defeat the purpose of the algorithm.

Ref.~\cite{kallaugher_how_2024} provides a framework for a hybrid quantum-classical algorithm to approximate problems that can lead to better algorithms than purely classical solutions.

\section{A Purely-Classical Streaming Algorithm for Signed Triangle Counting}\label{Sec:3}

Inspired by the optimal classical streaming JK algorithm proposed for (unsigned) edge stream in Ref.~\cite{jayaram2021optimal}, we develop a streaming algorithm for estimating number of triangles of type $T_1$ in a signed edge stream. This algorithm can be trivially adapted to count any (or all) $T_j$ types.

We first briefly recall the streaming JK algorithm. Prior to their work, several upper and lower bounds for triangle counting in insertion-only streams and for linear sketching algorithms had already been established; see Ref.~\cite{jayaram2021optimal} and the references therein. The key insight underlying their approach is the following: given three edges  $uv, vw, wu \in E$ arriving sequentially in a stream, if the first two edges $uv$ and $vw$ are sampled and stored, then upon observing the completing edge  $wu$, the presence of the triangle $uvw$ in the underlying graph $G$ can be detected.

However, naive random edge sampling alone generally leads to suboptimal performance. To overcome this limitation, Jayaram and Kallaugher introduce a carefully parameterized sampling strategy that combines both edge and vertex sampling. Specifically, edges are sampled with probability $p_E$ , while vertices are sampled independently with probability  $p_V$, and triangles are counted when they close sampled wedges. The sampling probabilities are chosen as

$$p_E = \frac{\Delta_V}{T}, \qquad 
p_V \ge \max\left\{ \frac{\Delta_V}{\Delta_E}, \frac{1}{\sqrt{\Delta_V}} \right\},$$
where $T$ denotes the total number of triangles, and  $\Delta_V$ and  $\Delta_E$ represent suitable concentration parameters as defined above. With this choice of parameters, the resulting algorithm achieves a space complexity of 
$$\mathcal{O}\!\left(\frac{m}{T}\left(\Delta_E + \sqrt{\Delta_V}\right)\right),$$
and yields an estimator with variance $\mathcal{O}(T^2)$ .

\begin{algorithm}[H]
        \caption{$T_1$ Estimator (Purely Classical)}\label{alg:T1_estimator_classical}
    \begin{algorithmic}[1]
        \Statex \textbf{Procedure} $\mathtt{T_1\_Estimator\_Subroutine\_FullyClassical(E; h_E,h_V;p_E,p_V)}$
        \Statex \textbf{Input}: $E$ signed edge stream;  $h_E,h_V$ probabilistic hashes; $p_E,p_V \in (0,1]$.
        \Statex \textbf{Output}: Estimate $R_1$.
        \vspace{1mm}
        \State $R_1 \gets 0, S \gets \{:\}$

        \For{$(vw,\sigma_{vw}) \in E$:}
            \For{$u \in V$:}
                \If{$h_V(u) = 1$ and $uv, uw \in S.keys()$} \Comment{Found a triangle, check it is a $T_1$}
                    \If{$\sigma_{vw} = +1$ and $S[uv], S[uw] = -1$}
                        \State $R_1 \gets R_1 + 1/\left( p_V p_E^2 \right)$
                    \ElsIf{$\sigma_{vw} = -1$ and $S[uv] \neq S[uw]$}
                        \State $R_1 \gets R_1 + 1/\left( p_V p_E^2 \right)$
                    \EndIf
                \EndIf
            \EndFor 
            \If{$h_E(vw) = 1$ and ($h_V(v) = 1$ or $h_V(w) = 1$)}
                \State $S[vw] \gets \sigma_{vw}$
            \EndIf
        \EndFor
        \State \textbf{return} $R_1$
    \end{algorithmic}
\end{algorithm}

Adapting a similar approach for signed edge stream, we consider same sampling probability for vertex and edges which are implemented using probabilistic hash functions denoted as $h_V$ and $h_E$ respectively. However, in signed edge stream we have for types of triangle $T_j,$ $j=0,1,2,3.$ To address the sign of the incoming edges in the stream, we store the edges in the sketch $S$ with sign. For an incoming edge $(vw,\sigma_{vw}),$ when a vertex $u$ is sampled and $uv, uw\in S$, we increase the count of the type of triangles $T_j$ based on $\sigma_{vw}$ the signs of $uv$ and $uw.$ The sketch is updated stochastically to include $vw$ depending on the probabilities $h_E(vw), h_V(v)$, and $h_V(w)$. We describe the algorithm for $j=1$ in Algorithm \ref{alg:T1_estimator_classical}. Note that our algorithm is an extension of the JK algorithm incorporating the signs of the edges in the stream and hence we call our algorithm signed JK algorithm.

Then we have the following corollary, which gives the bit complexity of space for the algorithm.

\begin{corollary}\label{cor:T1_median_of_means_purely_classical}
    For small error parameters $\epsilon, \delta > 0$, there is a median-of-means algorithm for insertion-only graph streams that approximates the number of $T_1$ triangles in a graph $G$ to $\epsilon T_1$ precision with probability $1-\delta$, using

    \begin{equation}
        \mathcal{O}\left( \frac{m}{T_1}\left(\Delta_E + \sqrt{\Delta_V}\right) \log n \frac{1}{\epsilon^2}\log \frac{1}{\delta} \right)
    \end{equation}
    
    \noindent bits of space. This algorithm uses \cref{alg:T1_estimator_classical} as a sub-routine.
\end{corollary}

As this is a simple signed extension of the JK algorithm, the proof follows almost identically to the proof of Theorem 1.1 in Ref.~\cite{jayaram2021optimal}.

\section{Hybrid Quantum-Classical Algorithm for Signed Triangle Counting}\label{sec:hybrid_alg}

Here we present the hybrid algorithm for estimating $T_1$ triangles. A similar approach can be adapted for other types of triangles. The entire framework is split into 2 subroutines that work in parallel as for estimating unsigned triangles~\cite{kallaugher_quantum_2022-1}. The quantum algorithm that estimates $T_{1}^{<k}$ is described in \cref{subsec:T1_q_estimator}. The classical algorithm that estimates $T_{1}^{>k}$ is described in \cref{subsec:c_estimator}. These are summed to provide an estimate on $T_1$, as described in \cref{subsec:hybrid_estimator}.

\subsection{Classical Estimator for $\Tgk_1$}\label{subsec:c_estimator}

We extend \texttt{ClassicalEstimator} from Ref.~\cite{kallaugher_quantum_2022-1} to the signed version in the following algorithm.

\begin{algorithm}[H]
    \caption{$\Tgk_1$ Estimator (Classical)}\label{alg:t1_greater_k_sub}
    \begin{algorithmic}[1]
        \Statex \textbf{Procedure}: $\mathtt{T_1\_Estimator\_Classical\_Subroutine(E; k; h_V,h_I)}$
        \Statex \textbf{Input}: $E$ signed edge stream ; $k \in [1,m]$; $h_V, h_I$ probabilistic hashes.
        \Statex \textbf{Output}: Estimate $\Rgk_1$.
        \vspace{1mm}
        \State $\Rgk_1 \gets 0, S \gets \emptyset, D_- \gets \{:\}, D_+ \gets \{:\}, \ell \gets 0$

        \For{$(vw,\sigma_{vw}) \in E$:}
            \For{$u \in V$:}
                \If{$\sigma_{vw} = +$}
                    \If{$(u,v,-), (u,w,-) \in S$} 
                        \State $\Rgk_1 \gets \Rgk_1 + 1-(1-1/k)^{D_-[(u,v)] +  D_+[(u,v)] + D_-[(u,w)] + D_+[(w,v)]}$
                    \EndIf   
                \Else
                    \If{$(u,v,+), (u,w,-) \in S$} 
                        \State $\Rgk_1 \gets \Rgk_1 + 1-(1-1/k)^{D_-[(u,v)] +  D_-[(u,w)] + D_+[(u,w)]}$
                    \EndIf 
                    \If{$(u,v,-), (u,w,+) \in S$} 
                        \State $\Rgk_1 \gets \Rgk_1 + 1-(1-1/k)^{D_-[(u,v)] +  D_+[(u,w)] + D_-[(u,w)]}$
                    \EndIf 
                \EndIf
                      
            \EndFor 
            \For{$(x,y,t) \in S$}
                \If{$y \in \{v,w\}$}
                    \If{$\sigma_{vw} = +1$}
                        \State $D_+[(x,y)] \gets D_+[(x,y)] + 1$
                    \Else
                        \State $D_-[(x,y)] \gets D_-[(x,y)] + 1$
                    \EndIf
                \EndIf
            \EndFor
            \If{$h_V(v) = 1$ and $h_I(\ell) = 1$}
                \State $S.add((v,w,\sigma_{vw}))$ and set $D_-[(v,w)], D_+[(v,w)] \gets 0$.
            \EndIf
            
            \If{$h_V(w) = 1$ and $h_I(\ell + 1) = 1$}
                \State $S.add((w,v,\sigma_{vw}))$ and set $D_-[(w,v)], D_+[(w,v)] \gets 0$.
            \EndIf
            \State $\ell \gets \ell + 2$
        \EndFor
        \vspace{2mm}
        \State \textbf{return} $\Rgk_1$
    \end{algorithmic}
\end{algorithm}

    The probabilistic hash functions must satisfy $Pr(h_V(u)) = 1/\sqrt{km}$ and $Pr(h_I(t)) = \sqrt{k/m}$. $h_V$ needs to be stored as it is called multiple times but it can be simplified to being only pairwise independent which requires $O(\log n)$ bits. $h_I$ only needs to be queried once and then forgotten, thus we do not need need to store the full function.

    \begin{lemma}[Analysis of \cref{alg:t1_greater_k_sub}]
        \begin{enumerate}
            \item In expectation, the algorithm requires uses $\mathcal{O}(\log n)$ bits to run.
            \item $\mathbb{E}(\Rgk_1) = \frac{\Tgk_1 \sqrt{k}}{m^{3/2}}$
            \item $Var(\Rgk_1) \leq 4 \frac{\Tgk_1 \sqrt{k}\Delta_E}{m^{3/2}}$
        \end{enumerate}
    \end{lemma}

    \begin{corollary}\label{lem:t1_greater_k}
        There is an median-of-means algorithm using \cref{alg:t1_greater_k_sub} as a subroutine that outputs an estimate of $\Tgk_1$ to $\epsilon T_1$ precision with probability $1-\delta$. This algorithm uses $\mathcal{O}\left(\frac{m^{3/2}}{T\sqrt{k}}\Delta_E \log n \frac{1}{\epsilon^2}\log \frac{1}{\delta}\right)$ number of classical and quantum bits in expectation.
    \end{corollary}
    
    The proof of these statements mirrors those of lemmas 8-11 in Ref.~\cite{kallaugher_quantum_2022-1}. In particular, the only difference is that in lines 5, 8, and 10, we check if there is a $T_1$ triangle centered on the wedge $(vw, \sigma_{vw})$. On line 7 in $\mathtt{ClassicalEstimator}$ from Ref.~\cite{kallaugher_quantum_2022-1}, they check if there is \textit{any} triangle centered on the wedge $vw$. In both cases, once the desired triangle is found, the appropriate value is added to the estimate.


\subsection{Quantum Estimator for $\Tlk_1$}\label{subsec:T1_q_estimator}

\begin{figure}[!t]
    \centering
    \includegraphics[width=1.0\textwidth]{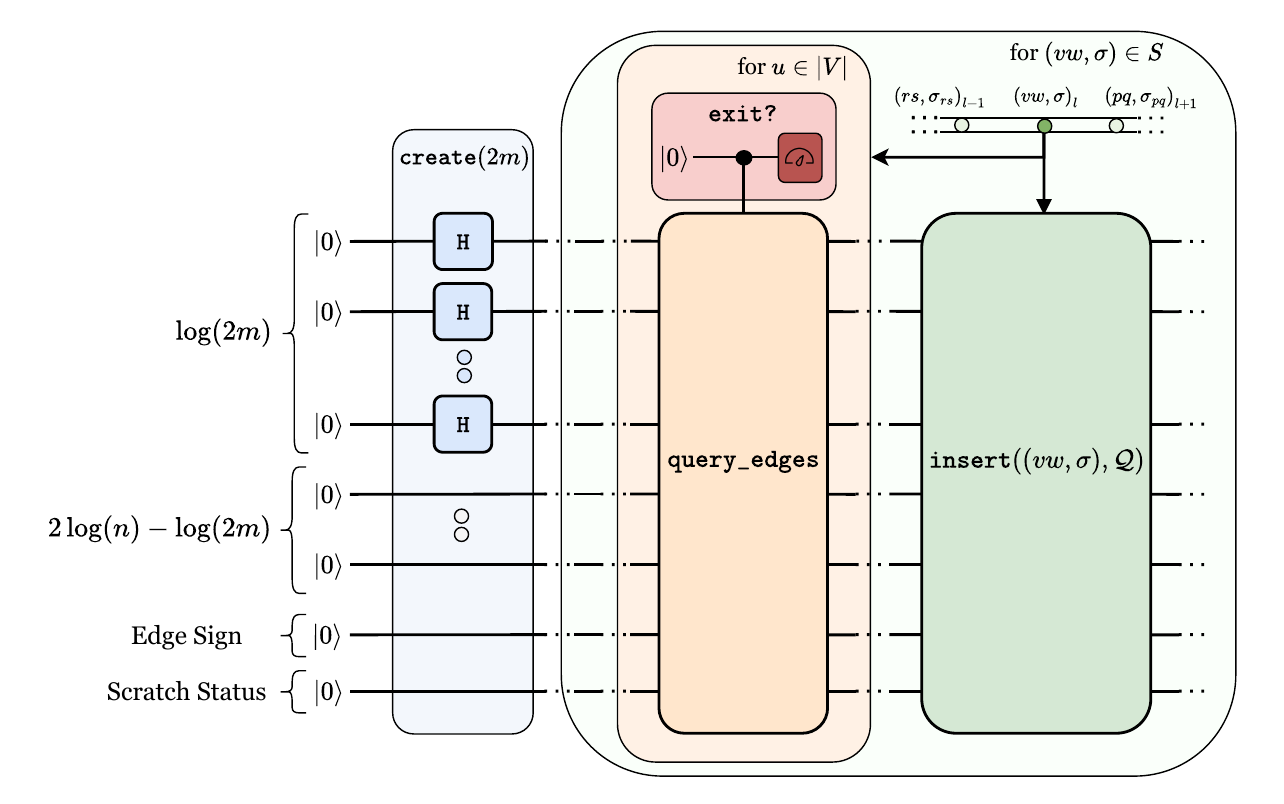}
    \caption{\textbf{Circuit Diagram for Quantum Sketching Algorithm.} The circuit for the quantum sketchpad is described by 3 types of action to update the quantum state as described in \cref{subsec:T1_q_estimator}. First, $\create$ is applied to generate the initial sketchpad, depicted is the case that $\log(2m)$ is a power of 2 such that Hadamard gates can be used. For each signed edge $(vw,\sigma)$ in the stream, $\queryedge$ is applied for every $u \in |V|$ related to $vw$ and $\sigma$. If one of the PVMs results in a measurement besides $\perp$, the algorithm halts and the value is returned. Lastly, $\Insert$ is applied to add $(vw,\sigma)$ and $(wv,\sigma)$ to the sketchpad and we proceed to the next edge.}
    \label{fig:q_sketch_circ}
\end{figure}

We now use a quantum subroutine to approximate $\Tlk_1$. The subroutine relies on 3 quantum primitives as outlined in Ref.~\cite{kallaugher_how_2024}. At a high level, this algorithm records a sketch $\Sc$ of seen edges into a quantum state $\QS$ and queries the sketch by measuring $\QS$. We maintain a state on $2\ceil{log_2 n} + 2$ qubits, where $n$ is the number of vertices in the underlying graph. Upon reading a signed edge $(vw, \sigma_{vw})$, a query is made on the underlying sketch for a wedge $wu, wv$ (along with sign information). This query is done via measurements (see \cref{sec:measurement_ops} for more details). If no wedge is found, then $(vw,\sigma_{vw})$ is inserted into the sketch and we proceed to the next edge. It is important to note that this algorithm is \textit{destructive}; processing an edge $(vw,\sigma_{vw})$ destroys previous information on edges incident to $v$ and $w$. A circuit corresponding to this is depicted in \cref{fig:q_sketch_circ}.

\begin{algorithm}[H]
        \caption{$\Tlk_1$ Estimator Subroutine}\label{alg:t1_less_k_sub}
    \begin{algorithmic}[1]
        \Statex \textbf{Procedure} $\mathtt{T_1\_Estimator\_Quantum\_Subroutine(E,k,g)}$
        \Statex \textbf{Input}: $E$ signed edge stream; $k \in [1,m]$; $g$ probabilistic hash.
        \Statex \textbf{Output}: Estimate $\Rlk_1$.
        \vspace{1mm}
        \State $\Qc_{\Sc} := \create([2m] \times \{\xi\})$. 
        \State $\ell := 0$
        \For{$(vw,\sigma_{vw}) \in E $:} \Comment{Assume $v < w$}
            \If{$g(\ell) = 1$}
            \For{$u \in V$:}
                \If{$\sigma_{vw} = +$}                     \State $r_{--} := \texttt{query\_edges}((uv, -), (uw, -), \Qc)$
                    \If{$r_{--} \ne \perp$}
                        \textbf{return} $r_{--} \cdot km$ 
                    \EndIf
                                                                                                \Else                                         \State $r_{+-} := \texttt{query\_edges}((uv, +), (uw, -), \Qc)$                     \If{$r_{+-} \ne \perp$}
                        \textbf{return} $r_{+-}\cdot km$
                    \EndIf
                                        \State $r_{-+} := \texttt{query\_edges}((uv, -), (uw, +), \Qc)$
                    \If{$r_{-+} \ne \perp$}
                        \textbf{return} $r_{-+}\cdot km$
                    \EndIf
                \EndIf
            \EndFor 
        \EndIf
        \State $\Insert((vw,s),\ell,  \Qc_{\Sc})$
        \State $\ell := \ell + 1$
        \EndFor
        \State \textbf{return} 0
    \end{algorithmic}
\end{algorithm}

In this extension of \texttt{QuantumEstimator} from Ref.~\cite{kallaugher_quantum_2022-1}, we need track of and query for signed edges. Note that the queries in lines 9 and 11 are on disjoint pairs and thus independent. 

\begin{lemma}[Performance of \cref{alg:t1_less_k_sub}]
    Let $\Rlk_1$ be the outcome of \cref{alg:t1_less_k_sub}.
    \begin{enumerate}
        \item The algorithm uses $2 \, \log(n) + 2 = O(\log n)$ qubits to run.
        \item $\mathbb{E}[\Rlk_1] = \Tlk_1$.
        \item $Var(\Rlk_1) \leq (km)^2$.
    \end{enumerate}
\end{lemma}

The proof of (2) and (3) from this algorithm are given in \cref{sec:analysis_q_est}.

\begin{corollary}\label{lem:t1_less_k}
    There is an median-of-means algorithm using \cref{alg:t1_less_k_sub} as a subroutine that outputs an estimate of $\Tgk_1$ to $\epsilon T_1$ precision with probability $1-\delta$. This algorithm uses $\mathcal{O}\left( \left( \frac{km}{T_1} \right)^2 \log (n) \frac{1}{\epsilon^2}\log{\left(\frac{1}{\delta}\right)} \right)$ number of classical and quantum bits in expectation. 
\end{corollary}

\subsection{Hybrid Algorithm for $T_{1} = \Tgk_1 + \Tlk_1$}\label{subsec:hybrid_estimator}
We are now ready to provide the full algorithm that estimates $T_1$. Note that this algorithm assumes that one has a reasonable bound on $m, \Delta_E,$ and $T_1$ itself. This is an almost identical statement and proof to Theorem 1 in Ref.~\cite{kallaugher_quantum_2022-1}.

\begin{theorem}
    For any $\epsilon, \delta \in (0,1]$, there is a hybrid quantum-classical streaming algorithm that uses

    \begin{equation}
        \mathcal{O}\left( \frac{m^{8/5}}{T_1^{6/5}}\Delta_E^{4/5} \log n \cdot \frac{1}{\epsilon^2}\log\frac{1}{\delta} \right)
    \end{equation}

    \noindent quantum and classical bits in expectation to return a $(1\pm \epsilon)$-multiplicative approximation of $T_1$ in an insertion-only graph stream with probability $1-\delta$.
\end{theorem}

\begin{proof}
    Let $\mathtt{T_1\_Estimator\_Classical}$ and $\mathtt{T_1\_Estimator\_Quantum}$ be the names of the median-of-means algorithms promised by \cref{lem:t1_greater_k,lem:t1_less_k}. The following algorithm achieves the performance of the theorem. Lines 2 and 3 are written sequentially but should really be run in parallel.

    \begin{algorithm}[H]
        \caption{Signed $T_{1}$ Estimator (Quantum-Classical)}\label{alg:t1_estimator_hybrid}
        \begin{algorithmic}[1]
            \Statex \textbf{Procedure} $\mathtt{T_1\_Estimator\_Hybrid(E, \epsilon, \delta)}$
            \Statex \textbf{Input}: $E$ stream of $m$ signed edges; $\epsilon, \delta$ error params.
            \Statex \textbf{Output}: Estimate $R_1$.
            \vspace{1mm}
            \State $k \gets \left \lceil \frac{(T_1)^{2/5}(\Delta_E)^{2/5}}{m^{1/5}}\right \rceil$
            \Statex Run the following sub-routines in parallel:
            \State $\Rgk_1 \gets \mathtt{T_1\_Estimator\_Classical(E, k, \epsilon/2, \delta/2)}$
            \State $\Rlk_1 \gets \mathtt{T_1\_Estimator\_Quantum(E, k, \epsilon/2, \delta/2)}$
            \State \textbf{return} $\Rgk_1 + \Rlk_1$
        \end{algorithmic}
    \end{algorithm}

    To show accuracy, recall that $\Rgk_1$ is an estimate of $\Tgk_1$ up to $\epsilon T_1/2$ precision with probability $1-\delta/2$. Similarly, $\Rlk_1$ is an estimate of $\Tlk_1$ up to $\epsilon T_1/2$ precision with probability $1-\delta/2$. Then $R_1 = \Rgk_1 + \Rlk_1$ is an estimate of $T_1 = \Tlk_1 + \Tgk_1$ to precision $\epsilon T_1$ with probability $1-\delta$ via a union bound.

    As for the space required, note that the estimates of $\Rgk_1$ and $\Rlk_1$ require $\tilde{\mathcal{O}}\left(\frac{m^{3/2}}{T_1\sqrt{k}}\Delta_E \right)$ and $\tilde{\mathcal{O}}\left(\left(\frac{km}{T_1} \right)^2 \right)$ space, respectively, ignoring the common log terms. When $k = \left \lceil \frac{(T_1)^{2/5}(\Delta_E)^{2/5}}{m^{1/5}}\right \rceil$, these both become $\tilde{\mathcal{O}}\left( \frac{m^{8/5}}{T_1^{6/5}}\Delta_E^{4/5} \right)$.
\end{proof}

Lastly, we would like to compare the space used by the hyrbid algorithm \cref{alg:t1_estimator_hybrid} versus its classical counterpart \cref{alg:T1_estimator_classical}. As noted in Ref.~\cite{kallaugher_quantum_2022-1}, the graph parameters that maximize the seperation are $\Delta_E = O(1)$ and $\Delta_V = \Omega(T_1) = \Omega(m)$, resulting in a $\tilde{O}(m^{2/5})$-space hybrid algorithm versus $\Omega(\sqrt{m})$-space purely classical algorithm. This is also true in our setting.

\section{A distributed  quantum-classical computing model for estimating balance of signed graphs}\label{Sec:5}

In this section, we estimate the triangular index of balance of signed graphs using the streaming algorithms developed in the previous section. In order to estimate the balance with accuracy $(1\pm \epsilon)$-multiplicative approximation with probability $1-\delta$ we estimate the counts of triangle types $T_j$ with accuracy $(1\pm \epsilon/(2+\epsilon))$-multiplicative approximation with probability $1-\delta.$ Using a simple algebraic manipulation it can be justified. Then we discuss a quantum distributed model for implementation of our proposed quantum-classical algorithm for estimating balance. 

\subsection{Algorithm for Estimating Triangular Balance}
We now use the individual triangle estimators to construct an algorithm that estimates triangle balance (\cref{eqn:defbalT}).

\begin{corollary}
    For any $\epsilon, \delta \in (0,1]$, there is a hybrid quantum-classical streaming algorithm that uses

    \begin{equation}
        \mathcal{O}\left( \frac{m^{8/5}}{T^{6/5}}\Delta_E^{4/5} \log n \cdot \frac{1}{\epsilon^2}\log\frac{1}{\delta} \right)
    \end{equation}

    \noindent quantum and classical bits in expectation to return a $(1\pm \epsilon)$-multiplicative approximation of $\mathcal{B}_{\triangle}(G)$ in an insertion-only graph stream with probability $1-\delta$.
\end{corollary}

Let $\mathtt{T\_Estimator\_Hybrid}$ be the algorithm described by Lemma 11 in Ref.~\cite{kallaugher_quantum_2022-1}.
    
    \begin{algorithm}[H]
        \caption{Balance Estimator (Hybrid)}\label{alg:bal_estimator_hybrid}
        \begin{algorithmic}[1]
            \Statex \textbf{Procedure} $\mathtt{BalanceEstimator\_Hybrid(E, \epsilon, \delta)}$
            \Statex \textbf{Input}: $S$ signed edge stream; $\epsilon, \delta \in (0,1]$ error params.
            \Statex \textbf{Output}: Estimate $\overline{\mathcal{B}}_{\triangle}^H(G)$.
            \vspace{1mm}
            \Statex Run the following sub-routines in parallel:
            \State $R_1 \gets \mathtt{T_1\_Estimator\_Hybrid(S, \epsilon/(2+\epsilon), \delta/3)}$
            \State $R_3 \gets \mathtt{T\_Estimator\_Hybrid(S, \epsilon/(2+\epsilon), \delta/3)}$\Comment{Dropping $-1$ edges}
            \State $R \gets \mathtt{T\_Estimator\_Hybrid(S, \epsilon/(2+\epsilon), \delta/3)}$ \Comment{Ignoring signs}

            \State \textbf{return} $(R_1 + R_3)/R$
        \end{algorithmic}
    \end{algorithm}

As a comparison to the hybrid algorithm, we also use a purely classical algorithm to estimate $\mathcal{B}_{\triangle}(G)$. This is a slight variation to \cref{alg:T1_estimator_classical}.

\begin{algorithm}[H]
        \caption{Balance Estimator (Purely Classical)}\label{alg:bal_estimator_classical}
    \begin{algorithmic}[1]
        \Statex \textbf{Procedure} $\mathtt{Bal\_Estimator\_Subroutine\_FullyClassical(E; h_E,h_V;p_E,p_V)}$
        \Statex \textbf{Input}: $E$ signed edge stream;  $h_E,h_V$ probabilistic hashes; $p_E,p_V \in (0,1]$.
        \Statex \textbf{Output}: Estimate $\overline{\mathcal{B}}_{\triangle}^C(G)$.
        \vspace{1mm}
        \State $R_{bal}, R_{unbal} \gets 0, S \gets \{:\}$

        \For{$(vw,\sigma_{vw}) \in E$:}
            \For{$u \in V$:}
                \If{$h_V(u) = 1$ and $uv, uw \in S.keys()$} \Comment{Found a triangle, check for correct type}
                    \If{$\sigma_{vw} = +$}
                        \If{($S[uv], S[uw] = -1$) or ($S[uv], S[uw] = +1$)}
                            \State $R_{bal} \gets R_{bal} + 1/\left( p_V p_E^2 \right)$ \Comment{Triangle type $T_1$ or $T_3$}
                        \Else
                            \State $R_{unbal} \gets R_{unbal} + 1/\left( p_V p_E^2 \right)$ \Comment{Triangle type $T_0$ or $T_2$}
                        \EndIf
                    \Else
                        \If{($S[uv], S[uw] = -1$) or ($S[uv], S[uw] = +1$)}
                            \State $R_{unbal} \gets R_{unbal} + 1/\left( p_V p_E^2 \right)$ \Comment{Triangle type $T_0$ or $T_2$}
                        \Else
                            \State $R_{bal} \gets R_{bal} + 1/\left( p_V p_E^2 \right)$ \Comment{Triangle type $T_1$ or $T_3$}
                        \EndIf
                    \EndIf
                \EndIf
            \EndFor 
            \If{$h_E(vw) = 1$ and ($h_V(v) = 1$ or $h_V(w) = 1$)}
                \State $S[vw] \gets \sigma_{vw}$
            \EndIf
        \EndFor
        \State \textbf{return} $R_{bal} / (R_{bal} + R_{unbal})$
    \end{algorithmic}
\end{algorithm}

\begin{conjecture}
    There is a polynomial space advantage to the hybrid quantum-classical triangle balance algorithm compared with state-of-the-art purely classical estimators.
\end{conjecture}

The motivation for the conjecture is that the hybrid algorithm is composed of 3 subroutines all with space advantage over their purely-classical counterparts. However, it is not clear if there is a better purely-classical algorithm to estimate balance that does not rely on explicitly calculating balanced and unbalanced triangles.

\subsection{The Distributed Computing Setup}

We now propose a distributed computing framework to implement the hybrid algorithm. Recall that distributed quantum computing (DQC) aims to scale quantum computation by leveraging quantum and/or classical communication links among multiple quantum devices~\cite{boschero2025distributed}. A central challenge in standard DQC architectures is the realization of distributed quantum operations, such as quantum gates acting across different devices. In contrast, our framework does not require any distributed quantum operations. Instead, each small quantum device operates independently, and the resulting classical outputs are communicated to a classical processor for post-processing to obtain the final estimate. This structure renders the proposed algorithm \emph{perfectly distributable} in the sense of Ref.~\cite{caleffi2024distributed}. Our approach is a form of embarrassingly parallelizable MapReduce paradigm, as depicted in \cref{fig:distri_schematic}. In the following, we describe the quantum and classical resources required and detail the implementation of the hybrid algorithm within this distributed setting.

\subsubsection{Computational resources and procedure}
For a fixed sign graph $G$ on $n$ vertices, let $\mathcal{B}_{\triangle}$ denote the true triangular balance. Then  given $0<\epsilon, \delta<1$ the hybrid algorithm estimates $\mathcal{B}_{\triangle}$ with precision $\epsilon \mathcal{B}_{\triangle}$ with probability $1-\delta$. Fix constants $N_Q = \Theta\left(\left(km/T\epsilon \right)^2\log \frac{1}{\delta} \right)$ and $N_C = \Theta\left(\frac{m^{3/2}\Delta_E}{T\sqrt{k}\epsilon^2}\log \frac{1}{\delta}\right)$. The DC setup requires $3N_Q$ small-scale QPUs and $3N_C$ CPUs. Here $m$ is an approximate number of edges in the stream, $T$ is the approximate number of triangles in $G,$ and $k\approx T^{2/5}\Delta_E^{2/5}/m^{1/5}.$ 

\begin{figure}[!t]
    \centering
    \includegraphics[width=0.90\linewidth]{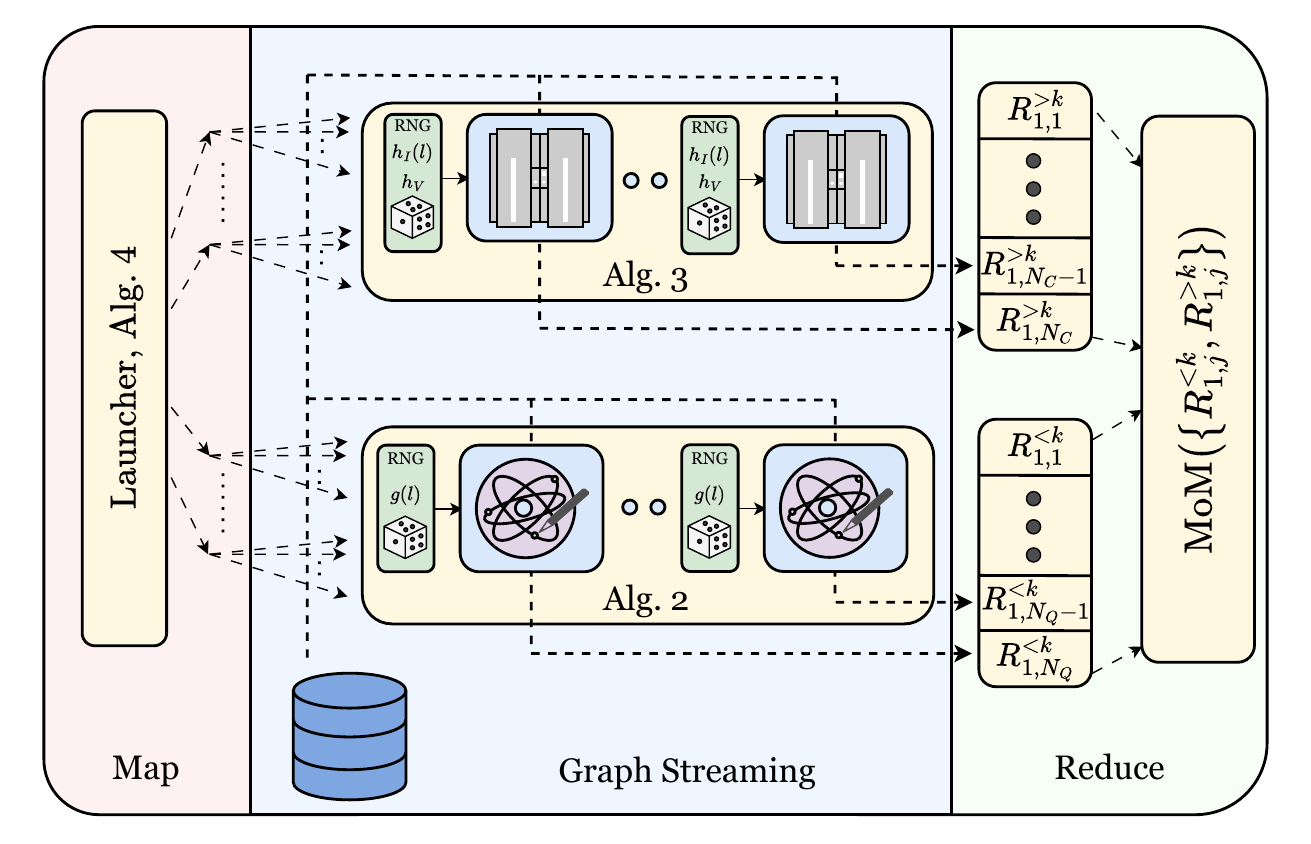}
    \caption{\textbf{A Schematic for Signed Triangle Estimation with Distributed Resources.} During the \textit{mapping}, based on Alg.~\ref{alg:t1_estimator_hybrid}, we allocate QPUs and CPUs with their respective initial (empty) sketchpads. During the \textit{Graph Streaming}, each edge is passed to these distributed CPUs and QPUs. At step $\ell$ in the stream, based on the hash functions $ g(\ell) $, $ h_I(\ell) $ and $h_V$, each CPU or QPU applies a \texttt{query} (see \cref{fig:q_sketch_circ}). Each CPU and QPU will terminate asynchronously and provide their estimate. During the \textit{reduction}, we calculate the median-of-means (MoM in diagram) based on parameters $\delta$ and $\epsilon$.}
    \label{fig:distri_schematic}
\end{figure}

\begin{enumerate}
    \item Quantum resources: Each QPU is supported with a classical register that implements the hash function $h_E$ independently for each incoming edge $e$ such that $h_E(e)=1$ with probability $1/k$ and $0$ otherwise. There is a   $(2\ceil{\log_2 n}+2)$-qubit quantum register, which implements the sketchpad. Indeed, the working qubit register is formed by $2\ceil{\log_2n}$ qubits, and there are two one-qubit registers used for scratch-qubit and the last one stores the sign of an incoming edge in the stream. Each QPU is also supported by a measurement device which performs queries to the quantum sketchpad as described by Algorithm \ref{alg:t1_less_k_sub}.

    \item Classical resources: Each CPU is supported by a hash function $h_V$ such that $\mathbb{E}[h_V(u)]=1/\sqrt{km}$ for each vertex $u$ of $G.$ In addition, it is associated with a sketchpad memory register which stores an incoming signed edge and probabilistically updated as described by Algorithm \ref{alg:t1_greater_k_sub}.
\end{enumerate}

Now we discuss the working procedure of the distributed computing setup as illustrated by a schematic diagram in Figure \ref{fig:distri_schematic}. The overall procedure approximates three values in parallel: $R_1$, $R_3$, and $R$. By outputting $(R_1 + R_3)/R$, we get an approximation of $\mathcal{B}_{\triangle}$.

A central classical organizer is in charge of by preparing $3N_Q$ QPUs and $3N_C$ CPUs. Each of these requires their own independent randomness and is in charge of calculating an estimate for one of the three values above. Roughly speaking, $N_Q$ QPUs and $N_C$ CPUs contribute to the approximation of $R_1$, $N_Q$ QPUs and $N_C$ CPUs contribute to the approximation of $R_3$, and $N_Q$ QPUs and $N_C$ CPUs contribute to the approximation of $R$. As edges are streamed, they are broadcast to every processor (quantum and classical) simultaneously. Once all the devices (quantum and classical) are done, they report their values to the central processor where classical \texttt{Median-of-Means} post-processing is performed to obtain the three individual estimates.

To make this process more precise, we focus on the $R_1$ approximation (see \cref{fig:distri_schematic}). For each $j \in [N_C]$, a CPU independently produces $\Rgk_{1,j}$ which is returned to the central organizer. Let $N_{C,\epsilon} := \frac{m^{3/2}\Delta_E}{T_1\sqrt{k}\epsilon^2}$ and $N_{C,\delta} := \log \frac{1}{\delta}$ such that $N_C = N_{C,\epsilon} \cdot N_{C,\delta}$. We bucket the first $N_{C,\epsilon}$ values to compute

\begin{equation}
    \widehat{\Rgk_{1,1}} := \mathtt{mean}\left( \Rgk_{1,1}, \dots, \Rgk_{1,N_{C,\epsilon}} \right).
\end{equation}

\noindent Next, compute $\widehat{\Rgk_{1,2}}$ by bucketing and taking the mean of the next $N_{C,\epsilon}$ values $\left\{\Rgk_{1,N_{C,\epsilon} + 1}, \dots, \Rgk_{1,2N_{C,\epsilon}}\right\}$. Proceed until we have averages $\widehat{\Rgk_{1,1}}, \dots, \widehat{\Rgk_{1,N_{C,\delta}}}$. From here, we attain the final classical estimate of $\Tgk_1$ as

\begin{equation}
    \Rgk_1 := \mathtt{median}\left(\widehat{\Rgk_{1,1}}, \dots, \widehat{\Rgk_{1,N_{C,\delta}}}\right).
\end{equation}

\noindent $\Rgk_1$ is then a good estimator for $\Tgk_1$.

A similar post-processing is done with the quantum outputs. First, divide the $N_Q$ QPUs in charge of estimating $\Rlk_1$ into $N_{Q,\epsilon} := \left(km/T_1\epsilon \right)^2$ buckets, each producing an average $\widehat{\Rlk_{1,j}}$ for $j \in [N_{Q,\delta}]$, where $N_{Q,\delta} := \log\frac{1}{\delta}$. Take the median of these values to produce $\Rlk_1$ which is an estimate of $\Tlk_1$. By summing $\Rgk_1$ and $\Rlk_1$, we obtain a good estimate of $T_1$. This same process is repeated in parallel to obtain estimates of $R_3$ and $R$. To finish, we output $(R_1 + R_3)/R$ to achieve an estimate $\overline{\mathcal{B}}_{\triangle}^H$ of $\mathcal{B}_{\triangle}$.

\subsection{Numerical Simulations}
\subsubsection{Simulating the Quantum Sketchpad}
In order to test the hybrid algorithm \cref{alg:bal_estimator_hybrid} on real graphs, we need to implement the quantum operations. The space required for any non-trivial example is much too large to be run on current devices so this process is simulated. This is done by a classical sketch whose operations obey the probabilities described in the \cref{subsec:T1_q_estimator}. In particular, the simulated sketch has the three methods $\create$, $\Insert$, $\queryedge$ that mimics the probabilistic behavior of the quantum sketch. See \cref{sec:python_code} for the code implementation.

\subsubsection{Results}
We run \cref{alg:bal_estimator_hybrid,alg:bal_estimator_classical} on a set of \ER graphs that additionally have randomly assigned edge signs. These graphs were created with two probabilities: the edge probability $p_{e}$ and for each edge, the probability that the sign is positive $p_+$. For each tuple of parameters  $(n,p_e,p_+) \in \{10, 20, 30, 40, 50\} \times  \{0.5, 0.75\}\times \{0.25, 0.5, 0.75\}$, 5 random graphs were created. \cref{tab:numerics_30,tab:numerics_40,tab:numerics_50} shows the $n \in \{30,40,50\}$ instances generated and used in these experiments. For conciseness, the $n \in \{10,20\}$ instances are omitted but the algorithms' performances are similar to the larger graph sizes.

Each algorithm was run with error parameters $\epsilon,\delta = 0.1$. For a fixed graph, let $\overline{\mathcal{B}}_{\triangle}^H$ and $\overline{\mathcal{B}}_{\triangle}^C$ be the outputs of \cref{alg:bal_estimator_hybrid,alg:bal_estimator_classical}, respectively, and $\mathcal{B}_\triangle$ the true value. Then using these parameters, both $\overline{\mathcal{B}}_{\triangle}^H$ and $\overline{\mathcal{B}}_{\triangle}^C$ approximate the $\mathcal{B}_\triangle$ up to $0.1\mathcal{B}_\triangle$ accuracy with probability at least $0.9$. As can be seen in \cref{fig:bal_estimator_performance}, both estimators outperform the $\epsilon = 0.1$ error threshold. Moreover, the hybrid algorithm performs moderately outperforms the purely-classical signed JK algorithm.

    \begin{figure}
        \centering
        \includegraphics[width=0.5\linewidth]{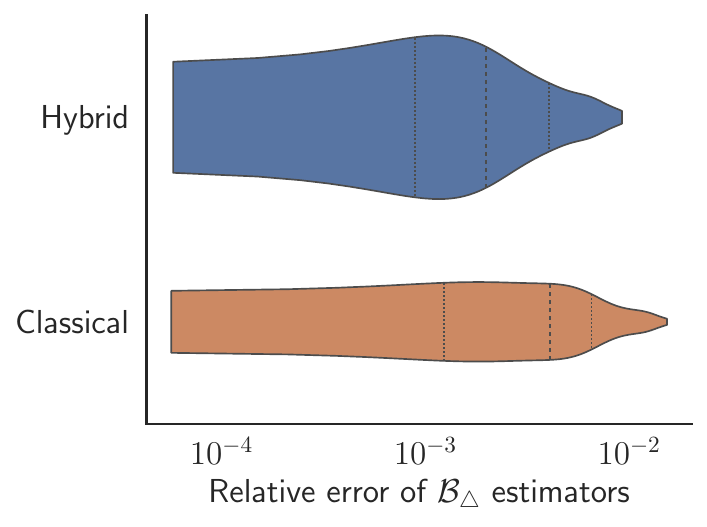}
        \caption{\textbf{Violin Plot Comparing the Performance of Classical and Quantum-Classical Triangular Balance Estimators.} Comparison between $\overline{\mathcal{B}}_{\triangle}^H$ and $\overline{\mathcal{B}}_{\triangle}^C$ outputted by \cref{alg:bal_estimator_hybrid} versus \cref{alg:bal_estimator_classical}, respectively. These algorithms are tested over random \ER instances of sizes $n \in \{30,40,50\}$. The relative error is computed as $|\overline{\mathcal{B}}_{\triangle}^A - \mathcal{B}_{\triangle}|/\mathcal{B}_{\triangle}$ for $A \in \{H,C\}$. Both algorithms are ran with accuracy parameter $\epsilon = 0.1$ and the performance is better than expected. The hybrid algorithm slightly outperforming the purely classical but does experience higher variance. This is to be expected, as there are two sources of error for the hybrid algorithm, on from the quantum and one from the classical estimators.}
        \label{fig:bal_estimator_performance}
    \end{figure}

\begin{table}
\caption{
            \textbf{Numerical Results for Signed Erd\H{o}s--R\'enyi Graphs with $n=30$ Nodes.}
            A bolded value corresponds to that estimate being closer to the true value
            $\mathcal{B}_{\triangle} = (T_1 + T_3)/\sum_i T_i$. In this table, the $\overline{\mathcal{B}}_{\triangle}^A$ values are truncated to three decimals here for compactness and so it may appear that they are equivalent for some rows.
            }
\label{tab:numerics_30}
\begin{tabular}{ccccccccccccc}
\toprule
$p_E$ & $p_+$ & $E_+$ & $E_-$ & $\Delta_E$ & $\Delta_V$ & $T_0$ & $T_1$ & $T_2$ & $T_3$ & $\mathcal{B}_{\triangle}$ & $\overline{\mathcal{B}}_{\triangle}^C$ & $\overline{\mathcal{B}}_{\triangle}^H$ \\
\midrule
0.75 & 0.5 & 148 & 171 & 21 & 233 & 215 & 676 & 557 & 147 & 0.516 & 0.513 & \textbf{0.517} \\
0.75 & 0.25 & 75 & 260 & 23 & 255 & 852 & 768 & 202 & 19 & 0.427 & 0.428 & \textbf{0.428} \\
0.5 & 0.25 & 42 & 156 & 13 & 79 & 185 & 165 & 26 & 5 & 0.446 & 0.445 & \textbf{0.445} \\
0.75 & 0.5 & 150 & 195 & 23 & 255 & 355 & 861 & 633 & 158 & 0.508 & 0.506 & \textbf{0.508} \\
0.5 & 0.5 & 108 & 111 & 13 & 101 & 72 & 199 & 196 & 61 & 0.492 & \textbf{0.492} & 0.492 \\
0.5 & 0.5 & 107 & 116 & 12 & 86 & 70 & 224 & 176 & 53 & 0.530 & 0.523 & \textbf{0.530} \\
0.5 & 0.25 & 62 & 161 & 14 & 104 & 221 & 261 & 77 & 14 & 0.480 & \textbf{0.479} & 0.479 \\
0.75 & 0.5 & 171 & 161 & 24 & 271 & 189 & 661 & 730 & 239 & 0.495 & 0.497 & \textbf{0.494} \\
0.5 & 0.25 & 51 & 164 & 12 & 80 & 221 & 202 & 63 & 5 & 0.422 & 0.420 & \textbf{0.421} \\
0.75 & 0.25 & 79 & 239 & 20 & 218 & 697 & 648 & 221 & 26 & 0.423 & 0.427 & \textbf{0.427} \\
0.5 & 0.5 & 101 & 110 & 12 & 86 & 71 & 174 & 171 & 53 & 0.484 & 0.481 & \textbf{0.482} \\
0.75 & 0.5 & 152 & 165 & 21 & 215 & 201 & 612 & 561 & 195 & 0.514 & \textbf{0.514} & 0.512 \\
0.5 & 0.5 & 94 & 114 & 12 & 72 & 63 & 173 & 138 & 42 & 0.517 & 0.522 & \textbf{0.520} \\
0.75 & 0.75 & 247 & 71 & 22 & 237 & 22 & 187 & 631 & 765 & 0.593 & \textbf{0.593} & 0.593 \\
0.5 & 0.25 & 50 & 168 & 15 & 103 & 230 & 241 & 57 & 10 & 0.467 & 0.472 & \textbf{0.464} \\
0.5 & 0.75 & 163 & 64 & 14 & 106 & 15 & 90 & 262 & 211 & 0.521 & 0.518 & \textbf{0.523} \\
0.5 & 0.75 & 150 & 57 & 11 & 70 & 10 & 86 & 167 & 164 & 0.585 & 0.585 & \textbf{0.585} \\
0.5 & 0.25 & 45 & 151 & 11 & 68 & 172 & 132 & 45 & 3 & 0.384 & 0.387 & \textbf{0.386} \\
0.75 & 0.25 & 100 & 233 & 24 & 267 & 632 & 795 & 345 & 51 & 0.464 & 0.468 & \textbf{0.465} \\
0.5 & 0.75 & 166 & 51 & 12 & 84 & 5 & 61 & 202 & 217 & 0.573 & 0.566 & \textbf{0.570} \\
0.75 & 0.75 & 254 & 78 & 23 & 284 & 28 & 257 & 689 & 813 & 0.599 & \textbf{0.599} & 0.598 \\
0.75 & 0.75 & 224 & 92 & 21 & 228 & 37 & 290 & 672 & 546 & 0.541 & 0.534 & \textbf{0.543} \\
0.75 & 0.5 & 170 & 137 & 20 & 240 & 122 & 470 & 569 & 248 & 0.510 & 0.513 & \textbf{0.508} \\
0.5 & 0.75 & 147 & 54 & 12 & 75 & 4 & 66 & 182 & 142 & 0.528 & 0.525 & \textbf{0.529} \\
0.75 & 0.75 & 255 & 76 & 22 & 247 & 22 & 226 & 720 & 816 & 0.584 & 0.587 & \textbf{0.583} \\
0.75 & 0.25 & 95 & 223 & 20 & 214 & 520 & 702 & 301 & 36 & 0.473 & \textbf{0.474} & 0.470 \\
0.75 & 0.75 & 255 & 76 & 23 & 263 & 23 & 217 & 739 & 819 & 0.576 & 0.572 & \textbf{0.574} \\
0.5 & 0.5 & 97 & 125 & 14 & 94 & 94 & 221 & 192 & 46 & 0.483 & \textbf{0.483} & 0.483 \\
0.75 & 0.25 & 87 & 244 & 23 & 245 & 719 & 799 & 264 & 23 & 0.455 & 0.456 & \textbf{0.456} \\
0.5 & 0.75 & 161 & 63 & 14 & 96 & 10 & 91 & 238 & 211 & 0.549 & 0.546 & \textbf{0.549} \\
\bottomrule
\end{tabular}
\end{table}
\begin{table}
\caption{
            \textbf{Numerical Results for Signed Erd\H{o}s--R\'enyi Graphs with $n=40$ Nodes.}
            A bolded value corresponds to that estimate being closer to the true value
            $\mathcal{B}_{\triangle} = (T_1 + T_3)/\sum_i T_i$. In this table, the $\overline{\mathcal{B}}_{\triangle}^A$ values are truncated to three decimals here for compactness and so it may appear that they are equivalent for some rows.
            }
\label{tab:numerics_40}
\begin{tabular}{ccccccccccccc}
\toprule
$p_E$ & $p_+$ & $E_+$ & $E_-$ & $\Delta_E$ & $\Delta_V$ & $T_0$ & $T_1$ & $T_2$ & $T_3$ & $\mathcal{B}_{\triangle}$ & $\overline{\mathcal{B}}_{\triangle}^C$ & $\overline{\mathcal{B}}_{\triangle}^H$ \\
\midrule
0.5 & 0.25 & 114 & 297 & 18 & 189 & 549 & 636 & 240 & 23 & 0.455 & \textbf{0.453} & 0.459 \\
0.75 & 0.75 & 443 & 156 & 31 & 480 & 74 & 669 & 1901 & 1810 & 0.557 & \textbf{0.557} & 0.560 \\
0.75 & 0.5 & 284 & 306 & 30 & 450 & 602 & 1641 & 1513 & 505 & 0.504 & 0.505 & \textbf{0.503} \\
0.75 & 0.5 & 301 & 298 & 30 & 455 & 530 & 1681 & 1694 & 562 & 0.502 & 0.503 & \textbf{0.503} \\
0.5 & 0.75 & 315 & 73 & 17 & 151 & 6 & 99 & 436 & 667 & 0.634 & \textbf{0.635} & 0.637 \\
0.5 & 0.75 & 274 & 104 & 16 & 156 & 19 & 204 & 475 & 424 & 0.560 & \textbf{0.560} & 0.563 \\
0.5 & 0.25 & 98 & 311 & 17 & 171 & 631 & 584 & 176 & 19 & 0.428 & 0.429 & \textbf{0.429} \\
0.5 & 0.25 & 99 & 287 & 16 & 145 & 520 & 479 & 176 & 23 & 0.419 & 0.421 & \textbf{0.419} \\
0.75 & 0.75 & 443 & 147 & 29 & 443 & 60 & 597 & 1796 & 1778 & 0.561 & 0.559 & \textbf{0.561} \\
0.75 & 0.5 & 293 & 286 & 30 & 439 & 503 & 1522 & 1460 & 543 & 0.513 & \textbf{0.512} & 0.513 \\
0.75 & 0.25 & 123 & 450 & 29 & 431 & 1898 & 1562 & 424 & 44 & 0.409 & 0.407 & \textbf{0.408} \\
0.5 & 0.5 & 217 & 199 & 18 & 177 & 168 & 508 & 583 & 193 & 0.483 & 0.481 & \textbf{0.482} \\
0.75 & 0.5 & 300 & 285 & 30 & 418 & 483 & 1502 & 1617 & 565 & 0.496 & 0.499 & \textbf{0.496} \\
0.75 & 0.75 & 440 & 138 & 29 & 395 & 39 & 528 & 1677 & 1716 & 0.567 & 0.568 & \textbf{0.566} \\
0.75 & 0.25 & 134 & 425 & 29 & 445 & 1587 & 1553 & 442 & 56 & 0.442 & 0.438 & \textbf{0.444} \\
0.75 & 0.25 & 155 & 428 & 28 & 424 & 1625 & 1789 & 656 & 82 & 0.451 & 0.454 & \textbf{0.448} \\
0.75 & 0.75 & 466 & 130 & 30 & 436 & 36 & 487 & 1775 & 2134 & 0.591 & \textbf{0.591} & 0.591 \\
0.5 & 0.5 & 186 & 198 & 18 & 166 & 167 & 416 & 458 & 129 & 0.466 & \textbf{0.466} & 0.465 \\
0.5 & 0.25 & 87 & 290 & 17 & 172 & 519 & 470 & 142 & 12 & 0.422 & 0.424 & \textbf{0.424} \\
0.5 & 0.75 & 282 & 102 & 17 & 160 & 17 & 180 & 521 & 448 & 0.539 & 0.544 & \textbf{0.542} \\
0.5 & 0.75 & 297 & 109 & 18 & 171 & 27 & 226 & 586 & 590 & 0.571 & 0.573 & \textbf{0.570} \\
0.5 & 0.5 & 189 & 208 & 18 & 160 & 172 & 506 & 476 & 128 & 0.495 & \textbf{0.493} & 0.496 \\
0.75 & 0.25 & 148 & 423 & 29 & 439 & 1584 & 1636 & 585 & 64 & 0.439 & 0.436 & \textbf{0.438} \\
0.5 & 0.25 & 92 & 317 & 19 & 200 & 644 & 576 & 174 & 19 & 0.421 & 0.419 & \textbf{0.421} \\
0.75 & 0.75 & 441 & 144 & 29 & 421 & 59 & 581 & 1743 & 1760 & 0.565 & 0.567 & \textbf{0.564} \\
0.5 & 0.5 & 197 & 195 & 16 & 167 & 159 & 465 & 485 & 150 & 0.488 & \textbf{0.489} & 0.493 \\
0.75 & 0.5 & 306 & 282 & 30 & 426 & 445 & 1545 & 1637 & 611 & 0.509 & \textbf{0.509} & 0.507 \\
0.75 & 0.25 & 164 & 421 & 29 & 440 & 1528 & 1796 & 719 & 96 & 0.457 & 0.460 & \textbf{0.457} \\
0.5 & 0.5 & 204 & 186 & 20 & 189 & 141 & 412 & 501 & 175 & 0.478 & \textbf{0.478} & 0.477 \\
0.5 & 0.75 & 291 & 94 & 15 & 130 & 14 & 141 & 504 & 514 & 0.558 & 0.561 & \textbf{0.556} \\
\bottomrule
\end{tabular}
\end{table}
\begin{table}
\caption{
            \textbf{Numerical Results for Signed Erd\H{o}s--R\'enyi Graphs with $n=50$ Nodes.}
            A bolded value corresponds to that estimate being closer to the true value
            $\mathcal{B}_{\triangle} = (T_1 + T_3)/\sum_i T_i$. In this table, the $\overline{\mathcal{B}}_{\triangle}^A$ values are truncated to three decimals here for compactness and so it may appear that they are equivalent for some rows.
            }
\label{tab:numerics_50}
\begin{tabular}{ccccccccccccc}
\toprule
$p_E$ & $p_+$ & $E_+$ & $E_-$ & $\Delta_E$ & $\Delta_V$ & $T_0$ & $T_1$ & $T_2$ & $T_3$ & $\mathcal{B}_{\triangle}$ & $\overline{\mathcal{B}}_{\triangle}^C$ & $\overline{\mathcal{B}}_{\triangle}^H$ \\
\midrule
0.75 & 0.5 & 453 & 479 & 37 & 683 & 1251 & 3250 & 3184 & 957 & 0.487 & \textbf{0.487} & 0.487 \\
0.5 & 0.5 & 301 & 297 & 21 & 229 & 271 & 819 & 894 & 304 & 0.491 & 0.494 & \textbf{0.494} \\
0.5 & 0.25 & 137 & 438 & 21 & 245 & 887 & 838 & 236 & 31 & 0.436 & \textbf{0.438} & 0.438 \\
0.5 & 0.75 & 444 & 160 & 20 & 217 & 39 & 348 & 1018 & 905 & 0.542 & 0.539 & \textbf{0.543} \\
0.75 & 0.5 & 453 & 460 & 37 & 669 & 1063 & 3020 & 3008 & 975 & 0.495 & \textbf{0.493} & 0.499 \\
0.5 & 0.25 & 152 & 456 & 22 & 232 & 988 & 1012 & 335 & 45 & 0.444 & 0.451 & \textbf{0.444} \\
0.75 & 0.25 & 222 & 697 & 35 & 610 & 3580 & 3447 & 1080 & 116 & 0.433 & 0.428 & \textbf{0.435} \\
0.75 & 0.5 & 441 & 473 & 33 & 610 & 1127 & 3191 & 2854 & 927 & 0.508 & 0.511 & \textbf{0.510} \\
0.5 & 0.5 & 313 & 310 & 21 & 221 & 281 & 986 & 961 & 315 & 0.512 & 0.515 & \textbf{0.509} \\
0.75 & 0.5 & 452 & 462 & 37 & 666 & 1034 & 3093 & 3039 & 925 & 0.497 & 0.502 & \textbf{0.496} \\
0.5 & 0.75 & 455 & 154 & 22 & 261 & 19 & 354 & 1041 & 981 & 0.557 & 0.558 & \textbf{0.557} \\
0.5 & 0.5 & 275 & 353 & 21 & 260 & 473 & 1104 & 839 & 220 & 0.502 & 0.505 & \textbf{0.504} \\
0.75 & 0.25 & 234 & 689 & 34 & 648 & 3462 & 3519 & 1160 & 152 & 0.443 & 0.442 & \textbf{0.443} \\
0.75 & 0.75 & 661 & 226 & 33 & 598 & 147 & 1038 & 3126 & 3120 & 0.560 & \textbf{0.559} & 0.559 \\
0.75 & 0.25 & 229 & 703 & 39 & 726 & 3684 & 3615 & 1195 & 108 & 0.433 & 0.429 & \textbf{0.434} \\
0.5 & 0.5 & 304 & 320 & 21 & 213 & 343 & 952 & 980 & 295 & 0.485 & 0.485 & \textbf{0.485} \\
0.5 & 0.25 & 148 & 455 & 21 & 234 & 1021 & 995 & 311 & 37 & 0.437 & 0.443 & \textbf{0.435} \\
0.75 & 0.5 & 449 & 472 & 37 & 706 & 1105 & 3179 & 3017 & 1001 & 0.503 & \textbf{0.506} & 0.508 \\
0.5 & 0.75 & 435 & 164 & 23 & 286 & 47 & 413 & 1023 & 893 & 0.550 & \textbf{0.550} & 0.547 \\
0.5 & 0.5 & 312 & 320 & 24 & 296 & 378 & 982 & 962 & 339 & 0.496 & 0.496 & \textbf{0.496} \\
0.5 & 0.75 & 450 & 153 & 20 & 225 & 37 & 328 & 1012 & 973 & 0.554 & 0.556 & \textbf{0.555} \\
0.5 & 0.25 & 134 & 466 & 21 & 220 & 1067 & 928 & 246 & 23 & 0.420 & 0.416 & \textbf{0.419} \\
0.75 & 0.25 & 238 & 665 & 36 & 647 & 3191 & 3331 & 1202 & 142 & 0.442 & 0.444 & \textbf{0.441} \\
0.5 & 0.25 & 155 & 457 & 20 & 252 & 1019 & 1017 & 336 & 39 & 0.438 & 0.438 & \textbf{0.438} \\
0.75 & 0.25 & 230 & 684 & 39 & 731 & 3429 & 3433 & 1192 & 126 & 0.435 & 0.436 & \textbf{0.435} \\
0.75 & 0.75 & 664 & 228 & 36 & 654 & 130 & 1091 & 3299 & 3075 & 0.549 & 0.549 & \textbf{0.549} \\
\bottomrule
\end{tabular}
\end{table}
\section{Conclusion}\label{sec:conclusion}

In this manuscript, we delineate a natural extension of the space-optimal classical algorithm for counting unsigned triangles into the signed case, where we need to count the four different types of signed triangles that can appear in a signed network stream. We then delineate a hybrid quantum-classical algorithm that provides a polynomial space advantage over the classical algorithm. 

Approximately counting the different types of triangles in a signed network has been tied to many kinds of important real-life network dynamics, through the lens of structural balance theory. In particular, we consider the task of estimating the balance of a sign network, tied to understanding conflict or discord in social networks, financial markets, biological systems, and more. 

Our work shows that quantum computers can be utilized in the monitoring and interpretation of large scale networks, achieving quantum advantage over the best possible classical algorithm for estimating central parameters such as balance. 

\bibliographystyle{IEEEtran}
\bibliography{refs}

\appendix

\section{Classical Simulation of Quantum Sketch}\label{sec:python_code}
Here we provide the Python code that implements the quantum sketch functionality.

\begin{lstlisting}[language=Python, caption=Python implementation of the classical simulation of the quantum Sketch funcionality.]
from random import random
class QuantumSketch:
    def __init__(self, items: set):
        self.items = items
        self.underlying = []

    def query_pair(self, x, y):
        if (x in self.items) and (y in self.items):
            if random() < 2.0 / len(self.items):
                self.items = []
                return +1
            else:
                self.items.remove(x)
                self.items.remove(y)
                return 0
            
        elif (x in self.items) or (y in self.items):
            if random() < 1.0 / len(self.items):
                self.items = []
                if random() < 0.5:
                    return -1
                else:
                    return +1
            else:
                self.items.discard(x)
                self.items.discard(y)
                return 0
        
        else:
            return 0
            
    def update(self, x, l):
        self.items.remove(l)
        self.items.add(x)
\end{lstlisting}

\section{Random Signed Graph Parameters Calculations}
Here we describe in more detail the graphs generated for the numerical experiments. We generate random \ER with $p_e$ the probability that an edge exists and $p_+$ the probability that a given edge has sign $+ 1$. These theoretical values can then be used when implementing and running the $T_1$ estimators, since these values are used in the number of resources required to run the algorithms themselves. For real-world graphs, these values can be estimated base on historical data. 

\subsection{$\abs{T_1}$ Expectation and Variance Calculations}

\begin{lemma}
    \begin{align}
        \Expect{\abs{T_1}} &= 3\binom{n}{3}p_e^3p_+(1-p_+)^2, \\
        \Var{\abs{T_1}} &= 3\binom{n}{3}p_e^3p_+(1-p_+)^2(1-3p_e^3p_+(1-p_+)^2) \\
        &\quad +2\binom{n}{2}\binom{n-2}{2}\left(p_e^5p_+(1-p_+)^3(1+3p_+)-9p_e^6p_+(1-p_+)^4\right)
    \end{align}
\end{lemma}

\begin{proof}
    Consider three distinct vertices $u,v,w \in V$. The probability they form a triangle is $p_e^3$. Let $I_{uvw}$ be the indicator variable that is 1 if $uvw$ is a $T_1$. The probability that the edges are a single $-1$ and two $+1$s is $p_+(1-p_+)^2$. There are three such ways for this to occur. So, $\prob{I_{uvw} = 1} = 3p_+(1-p_+)^2$. Noting that there are $\binom{n}{3}$ such triples, we get

    \begin{equation}
        \Expect{\abs{T_1}} = 3\binom{n}{3}p_e^3p_+(1-p_+)^2
    \end{equation}

    Onto the variance. Let $p_{\tau} := \prob{I_{\tau}}$ for a triple $\tau$. The variance can be calculated as 

    \begin{equation}
        \Var{\abs{T_1}} = \sum_{\tau} \Var{I_\tau} + 2\sum_{\tau < \tau'} \Cov{I_\tau, I_{\tau'}}
    \end{equation}

    \noindent Since $I_\tau$ is Bernoulli, we have $\Var{I_\tau} = p_{\tau}(1-p_{\tau}^2)$. For the covariance, let $uvw$ and $xyz$ be two triangles. We need to understand what happens when they share vertices.

    \begin{itemize}
        \item When $\{u,v,w\} \cap \{x,y,z\} = \emptyset$, these triangles are independent and the covariance is 0.
        \item When $\abs{\{u,v,w\} \cap \{x,y,z\}} = 1$, these triangles are independent and the covariance is 0.
        \item Assume $\abs{\{u,v,w\} \cap \{x,y,z\}} = 1$. Without loss of generality, they share edge $(u,v)$ and $w \neq z$. If $\sigma_{uv} = -1$, then the remaining 4 edges all must have $+1$ sign, which occurs with probability $p_+^4$. On the other hand, if $\sigma_{uv} = -1$, then 1 of $\sigma_{uw},\sigma_{vw}$ must be $+1$ and one must be $-1$, and the  same for $\sigma_{uz}, \sigma_{vz}$. Each occurs with probability $p_+^2(1-p_+)^2$. Putting these together, we get

        \begin{equation}
            p_+(1-p_+)^4 + 4p_+^2(1-p_+)^3 = p_+(1-p_+)^3(1+3p_+).
        \end{equation}

        \item  When $\abs{\{u,v,w\} \cap \{x,y,z\}} = 3$, these triangles are the  same and the covariance is 0.
    \end{itemize}

    \noindent The only non-zero case is when they share a single edge which allows us to calculate

    \begin{align}
        \prob{I_\tau, I_{\tau'} =1} &= p_e^5p_+(1-p_+)^3(1+3p_+) \\
        \Rightarrow \Cov{I_\tau, I_{\tau'}} &= p_e^5p_+(1-p_+)^3(1+3p_+) - p_{\tau}^2 \\
        &=p_e^5p_+(1-p_+)^3(1+3p_+) - 9p_e^6p_+^2(1-p_+)^4.
    \end{align}

    \noindent Note that there are $\binom{n}{2}\binom{n-2}{2}$ ways that the pair of triangles share an edge. We can now calculate the full variance:

    \begin{align}
        \Var{\abs{T_1}} &= \binom{n}{3}p_{uvw}(1-p_{uvw}) + 2\binom{n}{2}\binom{n-2}{2}[p_e^5p_+(1-p_+)^3(1+3p_+) - p_{uvw}^2] \\
        &=\begin{aligned}
            &3\binom{n}{3}p_e^3p_+(1-p_+)^2(1-3p_e^3p_+(1-p_+)^2) \\
            &+2\binom{n}{2}\binom{n-2}{2}\left(p_e^5p_+(1-p_+)^3(1+3p_+)-9p_e^6p_+(1-p_+)^4\right)
        \end{aligned}
    \end{align}
\end{proof}

\subsection{$\Delone$ Expectation Calculation}
        Next we are interested in understanding $\Delone$, the maximum number of $T_1$ triangles that share an edge. Let $Y_e$ be the number of $T_1$s that share edge $e$. Then $\Delone = \max_e Y_e$. Therefore, we study $\Delone$ via $Y_e$.

        Fix edge $e = (u,v)$. We can further define the indicator $Y_{uvw}$ to be 1 when $uvw$ forms a $T_1$. The probability that $uw$ and $vw$ are both edges is $p_e^2$. If $\sigma_{uv} = +1$, then both of the other signs must be $-1$, which occurs with probability $(1-p_+)^2$. Otherwise, one of the edges is $+1$ and $-1$. These individually occur with probability $p_+(1-p_+)$ and there are 2 ways for this to occur. Therefore:

        \begin{align}
            \prob{Y_{uvw} = 1} &= p_e^2\left[\overbrace{p_+}^{\sigma_{uv} = +1}(1-p_+)^2 + \overbrace{(1-p_+)}^{\sigma_{uv} = -1}\cdot 2p_+(1-p_+)\right] \\
            &=3p_e^2p_+(1-p_+)^2
        \end{align}

        \noindent Let $q_{uvw} := \prob{Y_{uvw} = 1}$. Since edges and their signs are independent, $Y_{uv} \sim Bin(n-2,q_{uvw})$ which gives

        \begin{equation}
            \Expect{Y_{uv}} = (n-2)q_{uvw}, \quad \Var{Y_{uv}} = (n-2)q_{uvw}(1-q_{uvw}).
        \end{equation}

        Next, let $M$ be the random variable for the number of edges in the graph. This is another Binomial variable of size $\binom{n}{2}$ with success probability $p_e$. In particular, $\Expect{M} \approx \frac{n^2}{2}p_e$.

        Binomials are sub-Gaussian, satisfying 

        \begin{equation}
            \prob{Y_e - \Expect{Y_e} > t} \leq \exp(-t^2/2\Var{Y_e})
        \end{equation}

        \noindent Therefore via a simple union bound:

        \begin{align}
            \prob{\max_{i = 1, \dots, M}Y_i \geq \Expect{Y_i} + t} = \prob{\bigcup_i \{Y_i - \Expect{Y_i} \geq t\} } &\leq \sum_i \prob{Y_e - \Expect{Y_e} > t} \\
            &\leq M\exp(-t^2/2\Var{Y_e}).
        \end{align}

        \noindent Letting $t = \sqrt{2\Var{Y_e} \log M}$ leads to the statement

        \begin{align}
            \Expect{\Delone} &\lesssim (n-2)q_{uvw} + \sqrt{2(n-2)q_{uvw}(1-q_{uvw})\log (n^2 p_e)}.\\
            & \nonumber
        \end{align}

\newcommand{\SKtodo}[1]{\ifthenelse{\boolean{showcomments}}{{\color{purple}{SK: #1}}}{}}

\newcommand{\peq}{\phantom{= }}
\newcommand{\Tpmm}{\Tlk_{+--}}
\newcommand{\Ent}[1]{\mathcal{E}_{NT}^{#1}}
\newcommand{\Et}{\mathcal{E}_{\mathtt{exit}}}
\newcommand{\Ew}[2]{\mathcal{E}_{\mathtt{wedge}, #1}^{#2}}
\newcommand{\Es}[2]{\mathcal{E}_{\mathtt{edge}, #1}^{#2}}

\newcommand{\vloop}{6}
\newcommand{\rmm}{8}
\newcommand{\rpm}{11}
\newcommand{\rmp}{13}
\newcommand{\supdate}{15}

\section{Quantum Streaming Primitives}
In this section, we outline the quantum streaming paradigm used in \cref{alg:t1_less_k_sub}. For convenience, we have renamed some methods used in Ref.~\cite{kallaugher_how_2024}; see their manuscript for a more general treatment of this paradigm.

\subsection{Quantum Registers}\label{sec:q_registers}
Let $n$ and $m$ be the number of vertices and edges in our graph. Our \textit{symbolic} universe for the sketch will be $U = \big\{ (u,v) \times \{+,-\} : u,v \in [n] \big\} \cup \big\{ (\ell, \xi) : \ell \in [2m] \big\}$, where $\xi$ denotes a scratch symbol. To be concrete and illustrative of the underlying logic, we consider a case where $n$ is a power of 2. Then the wavefunction state is decomposable as:
\begin{align}
\ket{\Qc} = \sum_{j=0}^{n-1} \sum_{k=0}^{n-1} \sum_{s=0}^{1} \sum_{t=0}^{1} \alpha_{jkst} \underbrace{ \ket{j}_{\log(n)} \otimes \ket{k}_{\log(n)} }_{\text{index register}} \otimes \underbrace{ \ket{s}_1  }_{\text{sign register}} \otimes \underbrace{ \ket{t}_1 }_{\text{scratch register}}, 
\end{align}
where a state $\ket{j}_{\log(n)}$ is the bit-wise decomposition $ j \in [n] $ over $\log(n)$ qubits in the computational basis with state normalization $ \sum_{jkst} \overline{\alpha_{jkst}} \alpha_{jkst} = 1$. 

We show how each element from our \textit{symbolic} universe maps to a concrete basis element of our quantum wavefunction. We define a scratch state as
\begin{align}\label{eq:defscratch}
\ket{\ell, \xi} \equiv \ket{q}_{\log(n)} \ket{r}_{\log(n)} \ket{0}_1 \ket{0}_1, 
\end{align}
where $ \ell = q \, n + r $ such that the quotient $ q \in [n] $ and the remainder $ r \in [n] $. 

We define an edge state as:
\begin{align}
\ket{u,v,\sigma_{uv}} \equiv \ket{u}_{\log(n)} \ket{v}_{\log(n)} \ket{\sigma_{uv}}_1 \ket{1}_1
\end{align}

Given noiseless application of all quantum operation as described in \cref{fig:q_sketch_circ}, assuming the algorithm has not exited yet, the state of the system can be described at the end of any timestamp $\ell$ by the following lemma.

\begin{lemma}[Sketchpad State After $\ell$ Edges]
\label{lm:triangle_stateinv}
For all $\ell = 0, \dots, m-1$, after the algorithm has processed $\ell$ edges, if it has not yet terminated the underlying set of $\Qc_{\Sc_{\ell}}$ is 
\begin{equation}	
\Sc_\ell := \set{(j, \xi) : j = 2\ell, \dots, 2m-1} \cup \Kc_\ell
\end{equation} 
where 
\begin{align}
    \Kc_\ell := \left\{(u, v, \sigma_{uv}) : \begin{aligned}
        &\Big( \exists i \in [\ell] \text{ s.t. } (u,v,\sigma_{uv}) = (e_i,\sigma_i) \Big) \bigwedge \\
        &\left( \underset{j=i+1,\dots,\ell}{\forall} \Big( g(j) = 0 \Big) \vee \Big( v \not \in e_j \Big) \vee \Big( \left( \sigma_{j} = + \right) \wedge \left(  \sigma_{uv} = + \right) \Big)\right)
    \end{aligned}\right\}
\end{align}
where $(e_i,\sigma_i)$ denotes the $i\nth$ signed edge in the stream. The size of $\Sc_\ell$ is $|\Sc_\ell| = 2m - 2\ell + |\Kc_\ell|$.

\end{lemma}

In particular, after processing $\ell$ edges from the stream, if the algorithm has not returned, the quantum state of the sketch is
\begin{align}
    \ket{\mathcal{Q}_{S_{\ell}}} = \frac{1}{\sqrt{2m - 2 \ell + |\Kc_\ell|}} \left( \sum_{j=2\ell}^{2m-1} \ket{j,\xi} + \sum_{(u,v,\sigma_{uv}) \in \Kc_\ell} \ket{u,v,\sigma_{uv}} \right) 
\end{align}

\noindent The set $\Kc_\ell$ contains all edges that we have seen that have not been deleted. Assume $(uv,\sigma_{uv})$ is an edge processed sometime before the $\ell\nth$ update. Consider timestep $j = i+1,\dots,\ell$. Then $(uv,\sigma_{uv})$ remains in the sketch iff:

\begin{enumerate}
    \item $g(j) = 0$ and thus no measurements are made at all.
    \item $v$ does not participate in the $j \nth$ edge and thus has no chance of being deleted.
    \item $v$ is in the $j \nth$ edge but both $\sigma_{uv},\sigma_j = +1$. This is due to the fact that the $j \nth$ set of measurements resulted in a $\bot$ but since $\sigma_j = +1$, we did not delete $+1$-neighbors.
\end{enumerate}

See \cref{pf:lm:triangle_stateinv} for a formal proof of this lemma.

\subsection{Measurement Operators}\label{sec:measurement_ops}

Here we describe the projective value measurements (PVMs) used to implement the $\mathtt{query\_edges}$ calls used by \cref{alg:t1_less_k_sub}. Recall that the quantum sketch is a superposition of at most $2m$ computational basis states. The states with non-zero amplitude can be described in two different ways. The first are  states $\ket{t,0,0}$ where  $t \in [2m]$ and the 3rd register denotes this is a ``scratch'' state. Secondly, we have states like $\ket{v,w,\pm, 1}$ where the first two registers $v,w \in V$ are vertices, the third register is $vw$'s sign $\sigma_{vw}$, and the fourth register denotes ``active'' status. This allows us to query for active states without worrying about conflicts of $t = (v,w)$ (as bitstrings).

Consider the scenario in which we are processing edge $vw$ with sign $+1$, leading to \textbf{step \rmm}. We want to query our sketch for a wedge $uv, uw$ such that both of the edges have $-1$ sign, thus completing a $T_1$ triangle with $uv$. This is accomplished by the $\queryedge((u,v,-), (u,w,-))$ call in \textbf{step \rmm} of Alg.~\ref{alg:t1_less_k_sub}. The following three operators form the PVM of this operation $\mathcal{M}_{uv-,uw-} := \{O^{+1}_{uv-,uw-}, O^{-1}_{uv-,uw-}, O^{\perp}_{uv-,uw-} \}$:
\begin{align}
    O^{+1}_{uv-,uw-} &:= \frac{1}{2} \Big( \ket{u,v,-} + \ket{u,w,-} \Big) \Big( \bra{u,v,-} + \bra{u,w,-} \Big) , \\
    O^{-1}_{uv-,uw-} &:= \frac{1}{2} \Big( \ket{u,v,-} - \ket{u,w,-} \Big) \Big( \bra{u,v,-} - \bra{u,w,-} \Big) ,  \\
    O^{\perp}_{uv-,uw-} &:= I - O^{+1}_{uv-,uw-} - O^{-1}_{uv-,uw-} . 
\end{align}

\noindent The observables $O^{+1}_{uv-,uw-}, O^{-1}_{uv-,uw-}$ are rank-1 projection operators onto two Bell-like states $\frac{1}{\sqrt{2}} \left( \ket{u,v,-} \pm \ket{u,w,-} \right) $. 

Consider the scenario in which we are processing edge $vw$ with sign $-1$, leading to \textbf{steps \rpm-\rmp}. We want to query our sketch for a wedge $uv, uw$ such that one of the edges has a $+1$ sign and the other $-1$, thus completing a $T_1$ triangle. This is accomplished by the $\queryedge((u,v,+), (u,w,-))$ call in \textbf{step \rpm} of Alg.~\ref{alg:t1_less_k_sub}. The following three operators form the PVM of this operation $\mathcal{M}_{uv+,uw-} := \{O^{+1}_{uv+,uw-}, O^{-1}_{uv+,uw-}, O^{\perp}_{uv+,uw-} \}$:
\begin{align}
    O^{+1}_{uv+,uw-} &:= \frac{1}{2} \Big( \ket{u,v,+} + \ket{u,w,-} \Big) \Big( \bra{u,v,+} + \bra{u,w,-} \Big) , \\
    O^{-1}_{uv+,uw-} &:= \frac{1}{2} \Big( \ket{u,v,+} - \ket{u,w,-} \Big) \Big( \bra{u,v,+} - \bra{u,w,-} \Big) , \\
    O^{\perp}_{uv+,uw-} &:= I - O^{+1}_{uv+,uw-} - O^{-1}_{uv+,uw-} . 
\end{align}

\noindent  The superscript in $O^{\pm 1}_{uv+,uw-}$ refers to this being a projection onto the superposition state $\frac{1}{\sqrt{2}} \left( \ket{u,v,+}\pm\ket{u,w,-}\right)$. The fourth register being set to 1 is used to only query for ``active'' data states and in particular, \textit{not} query for scratch states. The PVM on \textbf{step \rmp} of Alg.~\ref{alg:t1_less_k_sub} follows similarly, $\mathcal{M}_{uv-,uw+} := \{O^{+1}_{uv-,uw+}, O^{-1}_{uv-,uw+}, O^{\perp}_{uv-,uw+} \}$:
\begin{align}
    O^{+1}_{uv-,uw+} &:= \frac{1}{2} \Big( \ket{u,v,-} + \ket{u,w,+} \Big) \Big( \bra{u,v,-} + \bra{u,w,+} \Big) , \\
    O^{-1}_{uv-,uw+} &:= \frac{1}{2} \Big( \ket{u,v,-}-\ket{u,w,+} \Big) \Big( \bra{u,v,-} - \bra{u,w,+} \Big) , \\
    O^{\perp}_{uv-,uw+} &:= I - O^{+1}_{uv-,uw+} - O^{-1}_{uv-,uw+} .
\end{align}

\noindent Similarly to the PVMs of step \rmm, $O^{\pm1}_{uv-,uw+}$ are projections onto the superposition state $\frac{1}{\sqrt{2}} \left( \ket{u,v,-}\pm\ket{u,w,+}\right)$ respectively.

\subsection{Circuit Primitives}
We now explain the three quantum sub-circuits necessary for \cref{alg:t1_less_k_sub}: $\create, \Insert,$ and $\queryedge$.

The 3 quantum operations are detailed as follows.

\subsubsection{$\create$}
    This method instantiates the sketch by creating a superposition of $2m$ computational basis states. Here, $\Sc_0 := \set{(j, \xi) : j = 0, \dots, 2m-1}$ and we have 
    
    \begin{equation}
        \ket{\Qc_{\Sc_0}} = \frac{1}{\sqrt{2m}} \sum_{j=0}^{2m-1} \ket{j,\xi}
    \end{equation}
    
    \noindent When $m$ is a power of 2, this can be done simply with Hadamard; if not, this can still be done efficiently~\cite{efficient_superposition}. 

\subsubsection{$\Insert((vw,\sigma_{vw}), \ell, \QS)$}

    Stores $(vw,\sigma_{vw})$ into the sketch $\Sc$. This operation ``uses'' two scratch states by rotating them to $\ket{v,w,\sigma_{vw}}$ and $\ket{w,v,\sigma_{vw}}$. 
    
    Suppose we are in step $\ell$. We perform permutation $\pi_{\ell, vw, \sigma_{vw}}$ that swaps:
    \begin{itemize}
    	\item $(2\ell, \xi)$ for $(v, w, \sigma_{vw})$
    	\item $(2\ell+1, \xi)$ for $(w, v, \sigma_{vw})$ 
    \end{itemize}
    Since we process an edge on pair $\{v,w\}$ only once, this swap has the effect of mapping the state $\ket{2\ell, \xi}$ to a state $\ket{v,w,\sigma_{vw}}$ and the state $\ket{2\ell+1, \xi}$ to a state $\ket{w,v,\sigma_{vw}}$ which is equivalent to the mappings
    \begin{align}
    \ketbra{v,w,\sigma_{vw}}{2\ell, \xi} &\equiv \ketbra{v}{q_{2\ell}} \otimes \ketbra{w}{r_{2\ell}} \otimes \ketbra{\sigma_{vw}}{0} \otimes \ketbra{1}{0}, \\ 
    \ketbra{w,v,\sigma_{vw}}{2\ell+1, \xi} &\equiv \ketbra{w}{q_{2\ell+1}} \otimes \ketbra{v}{r_{2\ell+1}} \otimes \ketbra{\sigma_{vw}}{0} \otimes \ketbra{1}{0},
    \end{align}
    where $ 2 \ell = q_{2\ell} \, n + r_{2\ell} $ and $2\ell+1 = q_{2\ell+1} \, n + r_{2\ell+1} $.

    The entire unitary at this point (so as not to disturb the rest of the state) is simply identity on all other states:
    \begin{equation}
    U_{\pi_{\ell, vw,\sigma_{vw}}} =
        \begin{aligned}
        I &- \ketbra{2\ell,\xi}{2\ell,\xi} - \ketbra{v,w,\sigma_{vw}}{v,w,\sigma_{vw}} \\ 
            &+ \ketbra{v,w,\sigma_{vw}}{2\ell, \xi} + \ketbra{2\ell, \xi}{v,w,\sigma_{vw}}\\ 
            &- \ketbra{2\ell+1,\xi}{2\ell+1,\xi} - \ketbra{w,v,\sigma_{uv}}{w,v,\sigma_{wv}} \\ 
            &+ \ketbra{v,w,\sigma_{vw}}{2\ell+1, \xi} + \ketbra{2\ell+1, \xi}{v,w,\sigma_{vw}}  
        \end{aligned}
    \end{equation}

\subsubsection{$\queryedge$}\label{subsubsec:queryedge}
    The operation $\queryedge((wu,\tau_{wu}),(wv,\tau_{wv}))$ queries $\Sc$ for a wedge $(wu,wv)$ that would complete a triangle with $uv$ (along with edge information). This is done by making a measurement on $\QS$. There are three possible scenarios: (I) both signed edges are in the current sketch, (II) only one of them is in the sketch, and  (III) neither of them is in the sketch. Depending on which scenario is currently applicable, the following describes the possible outcomes with their associated probabilities.

    \begin{enumerate}
        \item[(I)] \begin{enumerate}
            \item[(a)] With probability $2/|\mathcal{S}|$, destroy $\QS$ and return +1.
            \item[(b)] With probability $1-2/|\mathcal{S}|$, the sketch is updated to $\Sc' = \Sc\setminus \{(wu,\tau_{wu}),(wv,\tau_{wv})\}$ and $\QS$ is updated to $\mathcal{Q}_{\mathcal{S'}}$. Return $\perp$.
        \end{enumerate}

        \item[(II)] \begin{enumerate}
            \item[(a)] With probability $1/|\mathcal{S}|$, destroy $\QS$ and with equal probability return $\{+1,-1\}$.
            \item[(b)] With probability $1-1/|\mathcal{S}|$, the sketch is updated to $\Sc' = \Sc\setminus \{(wu,\tau_{wu}),(wv,\tau_{wv})\}$ and $\QS$ is updated to $\mathcal{Q}_{\mathcal{S'}}$. Return $\perp$.
        \end{enumerate}
    
        \item[(III)] Leave $\QS$ unchanged and return $\perp$
    \end{enumerate}

    This query method is done by making measurements onto the quantum state $\ket{\QS}$ as outlined in \cref{sec:measurement_ops}. We work through the steps here for the measurement collection $\mathcal{M}_{wu-,wv+} = \{ O^{+1}_{wu-,wv+}, O^{-1}_{wu-,wv+}, O^{\perp}_{wu-,wv+} \}$ although similar logic can be used for $\mathcal{M}_{wu-,wv-}, \mathcal{M}_{wu+,wv-}$. 

    At the point of measurement, the wavefunction of the system is
    \begin{align}
    \ket{\QS} = \frac{1}{\sqrt{\abs{\Sc}}} \left( \sum_{j=2\ell}^{2m-1} \ket{j, \xi} + \sum_{(u,v,\sigma_{uv}) \in S} \ket{u,v,\sigma_{uv}} \right) 
    \end{align}
    
    \begin{enumerate}
    \item[(I)] \begin{enumerate}
    \item[(a)] With probability $2/|\mathcal{S}|$, destroy $\QS$ and return +1.
    \item[(b)] With probability $1-2/|\mathcal{S}|$, the sketch is updated to $\Sc' = \Sc\setminus \{(wu,\tau_{wu}),(wv,\tau_{wv})\}$ and $\QS$ is updated to $\mathcal{Q}_{\mathcal{S'}}$. Return $\perp$.
    \end{enumerate}

    In scenario $I$, $(w,u,-), (w,v,+) \in \mathcal{S}$. Then 
    \begin{align}
    \bra{\QS} O^{+1}_{wu-,wv+} \ket{\QS} &= \bra{\QS} \frac{1}{2} \Big( \ket{w,u,-} + \ket{w,v,+} \Big) \Big( \bra{w,u,-} + \bra{w,v,+} \Big) \ket{\QS} \\ 
    &= \frac{1}{2} \left( \frac{1}{\sqrt{\abs{\Sc}}} + \frac{1}{\sqrt{\abs{\Sc}}} \right) \left( \frac{1}{\sqrt{\abs{\Sc}}} + \frac{1}{\sqrt{\abs{\Sc}}} \right) \\ 
    &= \frac{2}{\abs{\Sc}},
    \end{align}
    and 
    \begin{align}
    \bra{\QS} O^{-1}_{wu-,wv+} \ket{\QS} &= \bra{\QS} \frac{1}{2} \Big( \ket{w,u,-} - \ket{w,v,+} \Big) \Big( \bra{w,u,-} - \bra{w,v,+} \Big) \ket{\QS} \\ 
    &= \frac{1}{2} \left( \frac{1}{\sqrt{\abs{\Sc}}} - \frac{1}{\sqrt{\abs{\Sc}}} \right) \left( \frac{1}{\sqrt{\abs{\Sc}}} - \frac{1}{\sqrt{\abs{\Sc}}} \right) \\ 
    &= 0,
    \end{align}
    and
    \begin{align}
    \bra{\QS} O^{\perp}_{wu-,wv+} \ket{\QS} &= \bra{\QS} I - O^{+1}_{wu-,wv+} - O^{-1}_{wu-,wv+} \ket{\QS} \\ 
    &= 1 - \frac{2}{\abs{\Sc}}. 
    \end{align}
    
    In the case that $O^{+}_{wu-,wv+}$ was observed, the wavefunction collapsed into the 1-dimensional space unambigiously. In the case $O^{\perp}_{wu-,wv+}$ was observed, the wavefunction is given by L{\"u}der's post-measurement update rule~\cite{luders1950zustandsanderung}:
    \begin{align}
    \frac{ O^{\perp}_{wu-,wv+} \ket{\QS} }{ \| O^{\perp}_{wu-,wv+} \ket{\QS} \| } &= \frac{1}{\| O^{\perp}_{wu-,wv+} \ket{\QS} \|}  \frac{1}{\sqrt{\abs{\Sc}}} \left( \sum_{j=2\ell}^{2m-1} \ket{j, \xi} + \sum_{(u,v,\sigma_{uv}) \in \Sc \setminus \{ (w,u,-),(w,v,+) \} } \ket{u,v,\sigma_{uv}} \right)  \\ 
    &= \frac{\sqrt{\abs{\Sc}}}{\sqrt{\abs{\Sc} - 2}} \frac{1}{\sqrt{\abs{\Sc}}} \left( \sum_{j=2\ell}^{2m-1} \ket{j, \xi} + \sum_{(u,v,\sigma_{uv}) \in \Sc \setminus \{ (w,u,-),(w,v,+) \} } \ket{u,v,\sigma_{uv}} \right) \\
    &= \frac{1}{\sqrt{\abs{\Sc}-2}} \left( \sum_{j=2\ell}^{2m-1} \ket{j, \xi} + \sum_{(u,v,\sigma_{uv}) \in \Sc \setminus \{ (w,u,-),(w,v,+) \} } \ket{u,v,\sigma_{uv}} \right)
    \end{align}
    
    \item[(II)] \begin{enumerate}
        \item[(a)] With probability $1/|\mathcal{S}|$, destroy $\QS$ and with equal probability return $\{+1,-1\}$.
        \item[(b)] With probability $1-1/|\mathcal{S}|$, the sketch is updated to $\Sc' = \Sc\setminus \{(wu,\tau_{wu}),(wv,\tau_{wv})\}$ and $\QS$ is updated to $\mathcal{Q}_{\mathcal{S'}}$. Return $\perp$.
    \end{enumerate}
    
    In scenario $II$, only one of $(w,u,-), (w,v,+)$ is in $ \mathcal{S} $. Then 
    \begin{align}
    \bra{\QS} O^{+1}_{wu-,wv+} \ket{\QS} &= \bra{\QS} \frac{1}{2} \Big( \ket{w,u,-} + \ket{w,v,+} \Big) \Big( \bra{w,u,-} + \bra{w,v,+} \Big) \ket{\QS} \\ 
    &= \frac{1}{2} \left( \frac{1}{\sqrt{\abs{\Sc}}} \right) \left( \frac{1}{\sqrt{\abs{\Sc}}} \right) \\ 
    &= \frac{1}{2\abs{\Sc}},
    \end{align}
    and 
    \begin{align}
    \bra{\QS} O^{-1}_{wu-,wv+} \ket{\QS} &= \bra{\QS} \frac{1}{2} \Big( \ket{w,u,-} - \ket{w,v,+} \Big) \Big( \bra{w,u,-} - \bra{w,v,+} \Big) \ket{\QS} \\ 
    &= \frac{1}{2} \left( \pm\frac{1}{\sqrt{\abs{\Sc}}} \right) \left( \pm\frac{1}{\sqrt{\abs{\Sc}}} \right) \\ 
    &= \frac{1}{2\abs{\Sc}},
    \end{align}
    and
    \begin{align}
    \bra{\QS} O^{\perp}_{wu-,wv+} \ket{\QS} &= \bra{\QS} I - O^{+1}_{wu-,wv+} - O^{-1}_{wu-,wv+} \ket{\QS} \\ 
    &= 1 - \frac{1}{\abs{\Sc}}. 
    \end{align}
    
    Similar to scenario $I$, if we observe $O^{+1}_{wu-,wv+}, O^{-1}_{wu-,wv+}$, the state has collapsed unambiguously into the 1-dimensional subspace. In the case $O^{\perp}_{wu-,wv+}$ is observed, by the same logic as for scenario $I$ we have that which ever element $(w,u,-)$ or $(w,v,+)$ present in $\mathcal{S}$ was deleted. 
    
    \item[(III)] Leave $\QS$ unchanged and return $\perp$. In this case it is clear that 
    \begin{align}
    \bra{\QS} O^{+1}_{wu-,wv+} \ket{\QS} &= 0 \\ 
    \bra{\QS} O^{-1}_{wu-,wv+} \ket{\QS} &= 0 
    \end{align}
    and so 
    \begin{align}
    \bra{\QS} O^{\perp}_{wu-,wv+} \ket{\QS} &= \bra{\QS} I - O^{+1}_{wu-,wv+} - O^{-1}_{wu-,wv+} \ket{\QS} \\ 
    &= 1. 
    \end{align}
    
    Since neither were in the set, the wavefunction is also unchanged. 
    
    \end{enumerate}

\section{Analysis of Quantum Estimator}\label{sec:analysis_q_est}
The main goal of this section is to prove \cref{lem:t1_less_k}. For clarity, we restate it here:

\begin{corollary}
\label{lem:t1_less_k_restated}
For any $\varepsilon, \delta \in (0,1\rbrack$, there is a quantum streaming
algorithm, using 
\begin{equation}	
\mathcal{O} \left( \frac{(km)^2 - (\Tlk_1)^2}{(\Tlk_1)^2}  \log n\frac{1}{\epsilon^2}\log\frac{1}{\delta} \right) 
\end{equation}

quantum and classical bits, that estimates $\Tlk_1$ to $\varepsilon T_1$
precision with probability $1-\delta$.
\end{corollary}

\noindent First, we need to show that the denominator matches (the square root of) the expectation of the algorithm:

\begin{lemma}\label{lem:Exp_X1}
    \begin{equation}
    \Expect{\Rlk_1} = \Tlk_1 
    \end{equation}
\end{lemma}

\noindent Next, we will need that the numerator is the upper-bound on the variance:

\begin{lemma}\label{lem:X1_var}
    \begin{equation}
    	\Var{\Rlk_1} \le (km)^2 - \left(\Tlk_1\right)^2 \le (km)^2
    \end{equation}
\end{lemma}

Using these two facts, we are able to show that a median-of-means algorithm gives the desired performance and whose space matches that of \cref{lem:t1_less_k}. These statements will be proven later.

\subsection{Quantum Sketch Updates}
In order to prove the desired statements above, we need to understand how the quantum sketch evolves as the algorithm progresses. For the rest of the section, let $g:[m]\to \bool$ be a fully independent hash function\footnote{For each $\ell \in [m]$ on step \textbf{5}, we receive $\log(n)$ random bits from our Random Number Generator.} such that 
\begin{equation}
    g(\ell) = \begin{cases}
    1 & \mbox{with probability $1/k$}\\
    0 & \mbox{otherwise.}
\end{cases}
\end{equation}

\begin{proof}[Proof (\cref{lm:triangle_stateinv})]\label{pf:lm:triangle_stateinv}
We proceed by induction on $\ell$. For $\ell = 0$, 
\begin{equation}	
\Sc_0 = \set{(j,\xi) \mid j \in [2m]}
\end{equation}
since we prepared the uniform state
\begin{align}
\ket{\Qc_{\Sc_{0}}} &= \left( \frac{1}{\sqrt{2m}} \sum_{j=0}^{2m-1} \ket{j}_{2\log(n)} \right) \ket{0}_1 \ket{0}_1 \\ 
&= \frac{1}{\sqrt{2m}} \sum_{j=0}^{2m-1} \ket{p_j}_{\log(n)} \ket{r_j}_{\log(n)} \ket{0}_1 \ket{0}_1 \\ 
&\equiv \frac{1}{\sqrt{2m}} \sum_{j=0}^{2m-1} \ket{j, \xi} & (\text{symbolic form, \cref{eq:defscratch})}
\end{align}
and so the result holds. We now prove by induction. For any $\ell \in \brac{m-1}$, suppose that the result
holds after $\ell$ edges.  Let $(x,y,\sigma_{xy})$ (with $x<y$) be the $(\ell + 1)\nth$
edge. 

First consider the effect of the $\mathtt{query\_edge}$ operations. If $g(\ell + 1) = 0$, they do not happen, and so $\Sc$ and thus $\QS$ are unchanged. If $g(\ell + 1) = 1$, we can assume all of them return $\bot$, as otherwise the algorithm would terminate.

Suppose $\sigma_{xy}=+$, then on \textbf{step \rmm} we apply the $\mathtt{query\_edges}$ associated with $(u,x,-),(u,y,-)$ for every $u \in V$. Then the associated state, in accordance with \cref{subsubsec:queryedge}, is:
\begin{align}
\ket{\Qc_{\Sc_{\ell+1}}} &= \frac{ O_{ux-,uy-}^{\perp} \ket{\Qc_{\Sc_{\ell}}}}{ \| O_{ux-,uy-}^{\perp} \ket{\Qc_{\Sc_{\ell}}} \| } \\ 
&= \frac{1}{ \| O_{ux-,uy-}^{\perp} \ket{\Qc_{\Sc_{\ell}}} \| } \left( I - O_{ux,-,uy,-}^{+1} - O_{ux-,uy-}^{-1} \right) \ket{\Qc_{\Sc_{\ell}}} \\ 
&= \frac{1}{ \| O_{ux,-,uy,-}^{\perp} \ket{\Qc_{\Sc_{\ell}}} \| } \left( \ket{\Qc_{\Sc_{\ell}}} - \Big( \bra{u,x,-}\ket{\Qc_{\Sc_{\ell}}} \Big) \ket{u,x,-} - \Big( \bra{u,y,-}\ket{\Qc_{\Sc_{\ell}}} \Big) \ket{u,y,-} \right) 
\end{align}

In particular, the braket inner products yield
\begin{align}
\bra{u,x,-}\ket{\Qc_{\Sc_{\ell}}} &= \begin{cases} 
\frac{1}{\sqrt{\Sc_\ell}} & (u,x,-) \in \Kc_\ell \\
0 & \text{otherwise}
\end{cases}, \\ 
\bra{u,y,-}\ket{\Qc_{\Sc_{\ell}}} &= \begin{cases} 
\frac{1}{\sqrt{\Sc_\ell}} & (u,y,-) \in \Kc_\ell \\
0 & \text{otherwise}
\end{cases}. 
\end{align}

So if $(u,x,-)$ or $(u,y,-)$ are in $\Sc_\ell$, they must have been deleted given that the measurement yielded $\perp$. Then after the loop of \textbf{step \vloop} finishes, we have deleted
\begin{align}
\Dc_{\ell+1}^{+} &= \set{(u,x,-) \in \Kc_\ell} \cup \set{(u,y,-) \in \Kc_\ell} \\ 
&= \set{(u,z,-) \in \Kc_\ell : z \in \{ x, y \} }
\end{align}
and so
\begin{align}
\Kc_{\ell+1} &= \Kc_\ell \setminus \Dc_{\ell+1}^{+} \cup \set{(x,y,+)} \cup \set{(y,x,+)}.
\end{align}
Then the underlying state after processing $\ell+1$ edges is
\begin{align}
\ket{\Qc_{\Sc_{\ell+1}}} &= \frac{1}{\sqrt{2m - 2 \ell + |\Kc_\ell| - |\Dc_{\ell+1}^{+}|}} \left( \sum_{j=2\ell+2}^{2m-1} \ket{j,\xi} + \ket{x,y,-} + \ket{y,x,-} + \sum_{u,v,\sigma_{uv} \in \Kc_\ell \setminus \Dc_{\ell+1}^{+}} \ket{u,v,\sigma_{uv}} \right) \\ 
&= \frac{1}{\sqrt{2m - 2 \ell - 2 + |\Kc_{\ell+1}|}} \left( \sum_{j=2\ell+2}^{2m-1} \ket{j,\xi} + \sum_{u,v,\sigma_{uv} \in \Kc_{\ell+1}} \ket{u,v,\sigma_{uv}} \right)
\end{align}

Suppose $\sigma_{xy}=-$, then on \textbf{step \rpm} we apply $\mathtt{query\_edges}$ associated with $(u,x,+),(u,y,-)$ for every $u \in V$. Then the associated state is:
\begin{align}
\ket{\Qc_{\Sc_{\ell+1}}} &= \frac{ O_{ux+,uy-}^{\perp} \ket{\Qc_{\Sc_{\ell}}} }{ \| O_{ux+,uy-} \ket{\Qc_{\Sc_{\ell}}} \| } \\ 
&= \frac{1}{ \| O_{ux+,uy-}^{\perp} \ket{\Qc_{\Sc_{\ell}}} \| } \left( I - O_{ux+,uy-}^{+1} - O_{ux+,uy-}^{-1} \right) \ket{\Qc_{\Sc_{\ell}}} \\ 
&= \frac{1}{ \| O_{ux-,uy-}^{\perp} \ket{\Qc_{\Sc_{\ell}}} \| } \left( \ket{\Qc_{\Sc_{\ell}}} - \Big( \bra{u,x,+}\ket{\Qc_{\Sc_{\ell}}} \Big) \ket{u,y,+} - \Big( \bra{u,y-1}\ket{\Qc_{\Sc_{\ell}}} \Big) \ket{u,y,-} \right)
\end{align}

In particular, the braket inner products yield
\begin{align}
\bra{u,x,+}\ket{\Qc_{\Sc_{\ell}}} &= \begin{cases} 
\frac{1}{\sqrt{\Sc_\ell}} & (u,x,+) \in \Kc_\ell \\
0 & \text{otherwise}
\end{cases}, \\ 
\bra{u,y,-}\ket{\Qc_{\Sc_{\ell}}} &= \begin{cases} 
\frac{1}{\sqrt{\Sc_\ell}} & (u,y,-) \in \Kc_\ell \\
0 & \text{otherwise}
\end{cases}. 
\end{align}

So if $(u,x,+)$ or $(u,y,-)$ are in $\Sc_\ell$, they must have been deleted given that the measurement yielded $\perp$. Then after the loop of \textbf{step \vloop} finishes, we have deleted
\begin{align}
\Dc_{\ell+1}^{+-} = \set{(u,x,+) \in \Kc_\ell} \cup \set{(u,y,-) \in \Kc_\ell},
\end{align}

Then on \textbf{step 13}, we similarly apply $\mathtt{query\_edges}$ associated with $(u,v,-),(u,y,+)$ for every $u \in V$. Then by a symmetric argument, the deleted set is
\begin{align}
\Dc_{\ell+1}^{-+} = \set{(u,x,-) \in \Kc_\ell} \cup \set{(u,y,+) \in \Kc_\ell},
\end{align}

The total deleted set after \textbf{steps \rpm-\rmp} is then
\begin{align}
\Dc_{\ell+1}^{-} &= \Dc_{\ell+1}^{-+} \cup \Dc_{\ell+1}^{+-} \\ 
&= \{ (u,x,\sigma_{ux}) \in \Kc_\ell \} \cup \{ (u,y,\sigma_{uy}) \in \Kc_\ell \} \\ 
&= \{ (u,z,\sigma_{uz}) \in \Kc_\ell : z \in \{x,y\} \} 
\end{align}

And so the set after adding $(x,y,-)$ and $(y,x,-)$ is
\begin{align}
\Kc_{\ell+1} = \Kc_\ell \setminus \Dc_{\ell+1}^{-} \cup \{ (x,y,-) \} \cup \{ (y,x,-) \} 
\end{align}

Thus, their effect on the underlying set of $\Qc$ can be summarized over both scenarios if $g(\ell+1) = 1 $ as to delete
\begin{align}
\Dc_{\ell+1} &= \begin{cases} 
\{ (u,z,-) \in \Kc_\ell : z \in \{x,y\} \} & \text{if } \sigma_{xy} = + \wedge g(\ell+1) = 1\\ 
\{ (u,z,\sigma_{uz}) \in \Kc_\ell : z \in \{x,y\} \} & \text{if } \sigma_{xy} = - \wedge g(\ell+1) = 1 \\ 
\emptyset & \text{otherwise}.
\end{cases} \\ 
&= \{ (u,z,\sigma_{uz}) \in \Kc_\ell : z \in \{x,y\}  \wedge (\sigma_{xy} = -1 \vee \sigma_{uz} = -1) \wedge (g(\ell+1) =1) \}
\end{align}
from it. In words, negative or positive edges incident on $xy$ are deleted if $\sigma_{xy}=-$ while only negative edges incident on $xy$ are deleted if $\sigma_{xy}=+$.

Then the updated set before and after calling $\Insert$ is
\begin{align}
\Kc_\ell \setminus \Dc_{\ell+1}
&= \left\{ (u, v, \sigma_{uv}) \in \Kc_\ell : ( g(\ell+1) = 0 ) \vee (v \notin xy) \vee (\sigma_{xy} = + \vee \sigma_{uv} = +) \right\}, \\
\Kc_{\ell+1} &= \Kc_\ell \setminus \Dc_{\ell+1} \cup \{ x, y, \sigma_{xy} \} \cup \{ y, x, \sigma_{xy} \},
\end{align}
which matches our expected lemma.

Specifically, $\exists i \in [\ell+1] \text{ s.t. } (u,v,\sigma_{uv}) = (e_i, \sigma_i) $ now includes $(x,y,\sigma_{xy}), (y,x,\sigma_{xy})$ from $i = \ell+1$ due to $\Insert((xy,\sigma_{xy}), \Qc$). The condition: 
\begin{align}
\underset{j=i+1,\dots,\ell+1}{\forall} \Big( g(j) = 0 \Big) \vee \Big( v \not \in e_j \Big) \vee \Big( \left( \sigma_{j} = + \right) \wedge \left(  \sigma_{uv} = + \right) \Big)
\end{align}
reflects that $\Dc_{\ell+1}$ is the empty set if $g(\ell+1)=$; if not, then $\Dc_{\ell+1}$ would include all negative or positive edges incident on $(e_{\ell+1},\sigma_{\ell+1})$ if $\sigma_{\ell+1}=-$ or all negative incident edges if $\sigma_{\ell+1}=+$.

\end{proof}

In order to analyze the algorithm, we need to argue about the underlying sketch at different points in time. Let $\Ec_{T}^{\ell}$ be the probability that the sketch was destroyed \textit{on} step $\ell$. Let $\Ec_{NT}^\ell$ be the event that the algorithm has not yet terminated by the end of processing the $\ell \nth$ edge. In particular, the quantum sketch still exists and has not collapsed and we have added the $\ell \nth$ edge to the sketch. Let $\Et^{\ell}$ be the event that we exit the algorithm while processing the $\ell \nth$ edge.
 
\begin{lemma}\label{lem:prob_not_terminated}
    For any $\ell \in [m]$,
\begin{align}
    \prob{\Ent{\ell}} &= \frac{|\Sc_{{\ell}}|}{2m},
\end{align}
\end{lemma}

This can be seen as an application of Lemma 1 in Ref.~\cite{kallaugher_how_2024}.

\begin{lemma}\label{lem:sketch_probs}
    Let $(vw,\sigma_{vw})$ be the $\ell \nth$ edge in a stream. Conditioned on $\Ent{\ell-1}$,

    \begin{itemize}
        \item[(a)] If $(uv,-) < (vw, \sigma_{vw})$, then the probability that $(uv,-)$ is still in the sketch is $(1-1/k)^{d_{uvw}}$.

        \item[(b)] If $(uv,+) < (vw, \sigma_{vw})$, then the probability that $(uv,+)$ is still in the sketch is $(1-1/k)^{d^-_{uvw}}$.
    \end{itemize}

    In particular, a positive-signed edge can only be deleted by negative-signed neighbors.
    
\end{lemma}

\begin{proof}
    (a) By the definition of $\Kc_{\ell-1}$, $(uv,-)$ is in the sketch iff it entered at some time $j < \ell$ and there were no signed edges incident to $v$ that arrived after $(uv,-)$ for an edge $j< k < \ell$ such that $g(k) = 1$. Suppose $vq$ is the $k \nth$ edge such that $j < k < \ell$ and $g(k) = 1$. If $\sigma_{vq} = -1$, then the $\mathtt{query\_edges}((uq, +), (uv, -), \Qc)$ call in line 11 (conditioned on it returning $\bot$) will delete $(uv,-)$ from the sketch. Similarly, $\sigma_{vq} = +1$, then the $\mathtt{query\_edges}((uq, -), (uv, -), \Qc)$ call in line 8 deletes $(uv,-)$ (conditioned on it returning $\bot$). Either way, conditioned on reaching $\ell$, the only way for $(uv,-)$ to still be in the sketch is for all incident edges to $v$ to come after $(uv,-)$ to not have triggered the query operations. The $j \nth$ set of queries occur with probability $1/k$, so the probability that all the queries do \textbf{not} occur for every such $j$ is given by $(1-1/k)^{d^+_{uvw} + d^-_{uvw}} = (1-1/k)^{d_{uvw}}$ since $g$ is fully independent.

    (b) By the definition of $\Kc_{\ell-1}$, $(uv,+)$ is in the sketch iff it entered at some time $j < \ell$ and there were no $\textbf{-1}$ signed edges incident to $v$ that arrived after $(uv,+)$ for an edge $j< k < \ell$ such that $g(k) = 1$. Suppose $vq$ is the $k \nth$ edge such that $j < k < \ell$ and $g(k) = 1$. If $\sigma_{vq} = -1$, then the $\mathtt{query\_edges}((uq, -), (uv, +), \Qc)$ call in line 13 (conditioned on it returning $\bot$) will delete $(uv,-)$ from the sketch. On the other hand, if $\sigma_{vq} = +1$, then the there is no way for $(uv,+)$ to be deleted from the sketch. Even if $g(k) = 1$, the $\mathtt{query\_edges}((uq, -), (uv, -), \Qc)$ call in line 8 as no bearing on $(uv,+)$. Therefore, we just need to track the $-1$ edges incident to $v$ and thus the probability that $(uv,+)$ is still in the sketch is $(1-1/k)^{d^-_{uvw}}$.
\end{proof}

\subsection{Signed Triangles in the Stream}
Fix an ordering on the stream. For edges $e,f$, we use $e < f$ to denote that $e$ arrives before $f$ in the stream. Fix vertices $u,v,w \in V$ such that $(uv,\sigma_{uv}) < (vw, \sigma_{vw})$. Define the positive-degree $d^+_{uvw}$ as the number of $+$-edges incident to $v$ that arrive between $(uv,\sigma_{uv})$ and $(vw, \sigma_{vw})$ (not including these edges themselves). Similarly define $d^-_{uvw}$. The full degree is then $d_{uvw} := d^+_{uvw} + d^-_{uvw}$.

\begin{align}
    \tlk_{uvw,+--} &= \begin{cases}
        (1 - 1/k)^{d^-_{uvw} + d_{uwv}} &\mbox{if $(uv,+) < (uw,-) < (vw,-)$ is a $T_1$ in $G$}\\
        0 &\mbox{otherwise.}
        \end{cases}\\
    \tlk_{uvw,-+-} &= \begin{cases}
        (1 - 1/k)^{d_{uvw} + d^-_{uwv}} &\mbox{if $(uv,-) < (uw,+) < (vw,-)$ is a $T_1$ in $G$}\\
        0 &\mbox{otherwise.}
        \end{cases}\\
    \tlk_{uvw,--+} &= \begin{cases}
        (1 - 1/k)^{d_{uvw} + d_{uwv}} &\mbox{if $(uv,-) < (uw,-) < (vw,+)$ is a $T_1$ in $G$}\\
        0 &\mbox{otherwise.}
        \end{cases}
\end{align}

\noindent We write $\Tlk_{+--} := \sum_{(u,v,w)\in V^3}\tlk_{uvw,+--}$. Similarly define $\Tlk_{-+-}, \Tlk_{--+}$ and let $\Tlk_1 = \Tlk_{+--} + \Tlk_{-+-} + \Tlk_{--+}$.

Likewise, define

\begin{align}
    \tgk_{uvw,+--} &= \begin{cases}
        1-(1 - 1/k)^{d^-_{uvw} + d_{uwv}} &\mbox{if $(uv,+) < (uw,-) < (vw,-)$ is a $T_1$ in $G$}\\
        0 &\mbox{otherwise.}
        \end{cases}\\
    \tgk_{uvw,-+-} &= \begin{cases}
        1-(1 - 1/k)^{d_{uvw} + d^-_{uwv}} &\mbox{if $(uv,-) < (uw,+) < (vw,-)$ is a $T_1$ in $G$}\\
        0 &\mbox{otherwise.}
        \end{cases}\\
    \tgk_{uvw,--+} &= \begin{cases}
        1-(1 - 1/k)^{d_{uvw} + d_{uwv}} &\mbox{if $(uv,-) < (uw,-) < (vw,+)$ is a $T_1$ in $G$}\\
        0 &\mbox{otherwise.}
        \end{cases}
\end{align}

\noindent and the corresponding $\Tgk_{+--}, \Tgk_{-+-}, \Tgk_{--+}, \Tgk_1$ values. Then $T_1 = \Tlk_1 + \Tgk_1$.

Recall that $\Rlk_1$ is the output of \cref{alg:t1_less_k_sub}. We break this random variable down further according to the different ways we can return from the algorithm. Consider the $\ell \nth$ edge $(vw,\sigma_{vw})$. We define the random variable $R_{u,vw;\sigma_{uv} \sigma_{uw} \sigma_{vw}}$ for any $u \in V$ and $\sigma_{uv}, \sigma_{uw} \in \{ \pm1 \}$ as follows. In any of the following cases, $R_{u,vw;\sigma_{uv} \sigma_{uw} \sigma_{vw}} = 0$: 

\begin{enumerate}
    \item If $g(\ell) = 0$ which means we do not make any measurements (and thus cannot return on this edge), or
    \item If any previous $\queryedge$ operation destroys the sketch (and thus we do not even reach edge $\ell$), or
    \item The value $r_{\sigma_{uv} \sigma_{uw}}$ returned by the measurement $\queryedge((u,v,\sigma_{uv}),(u,w, \sigma_{uw}))$ is equal to $\bot$.
\end{enumerate}

\noindent Otherwise, $R_{u,vw;\sigma_{uv} \sigma_{uw} \sigma_{vw}} = r_{\sigma_{uv} \sigma_{uw}}km$. Define

\begin{align}
    R_{--+} &= \sum_{(v<w) \in E}\sum_{u \in V} R_{u,vw;--+}\label{eq:X--+}, \quad \text{(from step 7)}\\
    R_{+--} &= \sum_{(v<w) \in E}\sum_{u \in V} R_{u,vw;+--}, \quad \text{(from step 10)} \label{eq:X+--}\\
    R_{-+-} &= \sum_{(v<w) \in E}\sum_{u \in V} R_{u,vw;-+-}. \quad \text{(from step 12)} \label{eq:X-+-}
\end{align}

\noindent Then $\Rlk_1 = R_{--+} + R_{-+-} + R_{+--}$. The explicit ordering of the nodes in an edge in the summation is to avoid double-counting issues. E.g. consider the triangle $(12, +) < (13, -) < (23, -)$. The algorithm also ensures that, even if the edge is read as $(32,-)$, since $2 < 3$, the $\queryedge((1,2,+),(1,3, -))$ call in line 10 is triggered. Therefore, we assume that $v < w$ moving forward.

Before proving \cref{lem:Exp_X1}, we need to understand the individual $\Expect{R_{u,vw;\sigma_{uv} \sigma_{uw} \sigma_{vw}}}$ values.

\begin{lemma}\label{lem:Exp_X--+}
    \begin{equation}
    \Expect{R_{u,vw;--+}} = \tlk_{uvw,--+} + \tlk_{uwv,--+}
    \end{equation}
\end{lemma}

Both values correspond to the triangle $\{(uv,-),(uw,-),(vw,+)\}$ in which $(vw,+)$ completes the triangle. Moreover, due to how these values are defined, one and only one can be non-zero but when it is, it is equal to $(1-1/k)^{d_{uvw} + d_{uwv}}$. This will be helpful in the proof below.

\begin{proof}
    First, consider the case where $\tlk_{uvw,--+} + \tlk_{uwv--+} = 0$. By the definition of $\tlk_{uvw,--+}$, this implies that either $\{(uv,-),(uw,-),(vw,+)\}$ is not a triangle in $G$, or that it is but at least one of $(uv,-)$ or $(uw,-)$ arrived after $(vw,+)$ in the stream. In either case, by Lemma~\ref{lm:triangle_stateinv}, at least one of $(uv,-)$ and $(uw,-)$ will not be in the sketch at the time the $(vw,+)$ queries occur, if they occur at all. If neither edge is in the sketch, the query is guaranteed to return $\bot$ and so $\Expect{R_{u,vw;--+}} = 0$. If exactly one of them is, then either the query returns $\bot$ or it is equally likely to return $1$ or $-1$, and so again $\Expect{R_{u,vw;--+}} = 0$. Either way, $\Expect{R_{u,vw;--+}} = 0 = \tlk_{uvw,--+} + \tlk_{uwv--+}$.

    Now suppose $\tlk_{u,vw--+} + \tlk_{u,wv--+} > 0$. Note that at most, one of these terms can be non-zero. By definition, $\tlk_{u,vw--+} + \tlk_{u,wv--+} = (1-1/k)^{d_{uvw} + d_{uwv}}$ no matter which term is non-zero. Without less of generality, let $\tlk_{u,vw--+} > 0$, so $(uv,-) < uw- < vw+$. We show that $\Expect{R_{u,vw;--+}} = (1-1/k)^{d_{uvw} + d_{uwv}}$, thus showing that $\Expect{R_{u,vw;--+}} = \tlk_{u,vw--+} + \tlk_{u,wv--+}$.

    Say that $(vw,+)$ is the $\ell\nth$ edge to arrive. Let $\mathcal{G}_\ell$ be the event that $g(\ell) = 1$ and recall that $\Ent{\ell-1}$ is the event that we have not yet terminated at the beginning of processing $(vw,+)$. Also, the underlying sketch at this time is $\Sc_{\ell-1}$ by Lemma~\ref{lm:triangle_stateinv}. Let $\Fc_{--}$ be the event that $(uv,-)$ and $(uw,-)$ are still in $\Sc_{\ell-1}$ conditioned on $\Ent{\ell-1}$. Since $g$ is fully independent, we have that 

    \begin{align}
        \prob{\Fc_{--}|\Ent{\ell-1}} &= \prob{(u,v,-) \in \Sc_{\ell-1}} \cdot \prob{(u,w,-) \in \Sc_{\ell-1}} \\ 
        &= (1-1/k)^{d_{uvw}} \cdot (1-1/k)^{d_{uwv}} \\ 
        &= (1-1/k)^{d_{uvw} + d_{uwv}} 
    \end{align}

    \noindent where the probabilities are due to \cref{lem:sketch_probs}. 
    
    Conditioned on $\Fc_{--}$ and $\Ent{\ell-1}$, $\queryedge$ returns +1 with probability $2/\abs{T_{l-1}}$ and $\perp$ otherwise. When $\perp$ is returned, $R_{u,vw;--+} = 0$ by definition. Then
    \begin{align}
        \Expect{R_{u,vw;--+} | \Fc_{--}, \Ent{\ell-1}, \mathcal{G}^\ell} =km \cdot \frac{2}{\abs{\Sc_{\ell-1}}}  + 0 \cdot \left(1-\frac{2}{\abs{\Sc_{\ell-1}}}\right)  = \frac{2km}{\abs{\Sc_{\ell-1}}} 
    \end{align}

    By Lemma~\ref{lem:prob_not_terminated} and the random choice of $g$,
    \begin{align}
        \prob{\mathcal{G}^\ell,\Ent{\ell-1}} = \prob{\mathcal{G}^\ell} \cdot \prob{\Ent{\ell-1}} = \frac{1}{k} \cdot \frac{\abs{\Sc_{\ell-1}}}{\abs{T_0}} =  \frac{\abs{\Sc_{\ell-1}}}{2km}.
    \end{align} 

    \noindent Also note that if we do not make any queries, then $R_{u,vw;--+} = 0$, so $\Expect{R_{u,vw;--+} | \overline{\mathcal{G}^\ell}} = 0$. Using this fact along with the law of total expectation:
    
    \begin{align}
                \Expect{R_{u,vw;--+} | \Fc_{--}} &= \Expect{R_{u,vw;--+} |  \Fc_{--}, \Ent{\ell-1}, \mathcal{G}^\ell} \prob{\Ent{\ell-1}, \mathcal{G}^\ell|\Fc_{--}}  \\
        &=\Expect{R_{u,vw;--+} |  \Fc_{--}, \Ent{\ell-1}, \mathcal{G}^\ell} \prob{\Ent{\ell-1}, \mathcal{G}^\ell} & (\text{Independent from } \Fc_{--})\\
        &=\frac{2km}{\abs{\Sc_{\ell-1}}} \cdot \frac{\abs{\Sc_{\ell-1}}}{2km} \\
        &= 1.
    \end{align}
    
    On the other hand, conditioning on $\overline{\Fc_{--}}$, there are 2 possibilities: neither edge is in the sketch or one edge is in the sketch. In case 1, the query returns $\bot$ with full probability. In case 2, the query returns $+1$ and $-1$ with equal probability, canceling each other out in expectation. Either way, $\Expect{R_{u,vw;--+} | \overline{\Fc_{--}}} = 0$. 

    So we conclude that 
    \begin{align}
        \Expect{R_{u,vw;--+}} &= \Expect{R_{u,vw;--+} | \Fc_{--}} \prob{\Fc_{--}} \\
        &= 1\cdot (1 - 1/k)^{\degb{uvw} + \degb{uwv}} \\
        & = \tlk_{uvw,--+} + \tlk_{uwv,--+}.
    \end{align}

    Therefore, no matter whether $\tlk_{uvw,--+} + \tlk_{uwv,--+}$ is 0 or not, we have $\Expect{R_{u,vw;--+}} = \tlk_{uvw,--+} + \tlk_{uwv,--+}$.
\end{proof}

\begin{lemma}\label{lem:Exp_X+--}
    \begin{equation}
    \Expect{R_{u,vw;+--}} = \tlk_{uvw,+--} + \tlk_{uwv,-+-} 
    \end{equation}
\end{lemma}

There is some subtlety in the subscripts on this statement but both values correspond to the triangle $\{(uv,+),(uw,-),(vw,-)\}$ in which $(vw,-)$ completes the triangle. Moreover, due to how these values are defined, one and only one can be non-zero but when it is, it is equal to $(1-1/k)^{d^-_{uvw} + d_{uwv}}$. This will be helpful in the proof below.

\begin{proof}[Proof]
    This proof is almost identical to that of \cref{lem:Exp_X+--}. The main difference is that we now want the event $\Fc_{+-}$ to be that $(uv,+)$ and $(uw,-)$ are still in $\Sc_{\ell-1}$ conditioned on $\Ent{\ell-1}$. Using the probabilities from probabilities \cref{lem:sketch_probs}, we get that
    \begin{align}
        \prob{\Fc_{+-}|\Ent{\ell-1}} &= \prob{(u,v,+) \in \Sc_{\ell-1}} \cdot \prob{(u,w,-) \in \Sc_{\ell-1}} \\ 
        &= (1-1/k)^{d^-_{uvw}} \cdot (1-1/k)^{d_{uwv}} \\ 
        &= (1-1/k)^{d^-_{uvw} + d_{uwv}} 
    \end{align}

    \noindent This in turn gives $\Expect{R_{u,vw;+--}} = (1-1/k)^{d^-_{uvw} + d_{uwv}} $ (when it is non-zero) and thus we have $\Expect{R_{u,vw;+--}} =\tlk_{uvw,+--} + \tlk_{uwv,-+-} $.
    
\end{proof}

\begin{lemma}\label{lem:Exp_X-+-}
    \begin{equation}
    \Expect{R_{u,vw;-+-}} = \tlk_{uvw,-+-} + \tlk_{uwv,+--}.
    \end{equation}
\end{lemma}

By symmetry, this proof is identical to that of \cref{lem:Exp_X+--}.

We are now able to prove that $\Expect{\Rlk_1} = \Tlk_1$.
\begin{proof}[(Proof of \cref{lem:Exp_X1})]
    Recall that $\Tlk_1 = \Tlk_{+--} + \Tlk_{-+-} + \Tlk_{--+}$. The goal is to show that the $\tlk$ values will be counted once and only once by the algorithm. First observe that

    \begin{align}
        \Expect{R_{--+}} &= \sum_{(v < w) \in E}\sum_{u \in V} \Expect{R_{u,vw;--+}}  \\
    &= \sum_{(v < w) \in E}\sum_{u \in V}\tlk_{uvw,--+} + \tlk_{uwv,--+} & \text{(\cref{lem:Exp_X--+})}\\
        &=\sum_{(u,v,w) \in V^3}\tlk_{uvw,--+} \\
        &=\Tlk_{--+}
    \end{align}

    \noindent where the second-to-last line follows form the fact that for a triple $(u,v,w) \in V^3$, one and only one ordering may lead to a non-zero $\tlk_{uv,w;--+}$ value.

    We must be slightly more careful with the remaining cases. Similar to above, we can use linearity of expectation to observe the following:

    \begin{align}
        \Expect{R_{+--}} &= \sum_{(v < w) \in E}\sum_{u \in V} \Expect{R_{u,vw;+--}} = \sum_{(v < w) \in E}\sum_{u \in V}\tlk_{uvw,+--} + \tlk_{uwv,-+-}   & \text{(\cref{lem:Exp_X+--})} \\
        \Expect{R_{-+-}} &= \sum_{(v < w) \in E}\sum_{u \in V} \Expect{R_{u,vw;-+-}}  = \sum_{(v < w) \in E}\sum_{u \in V}\tlk_{uvw,-+-} + \tlk_{uwv,+--}  & \text{(\cref{lem:Exp_X-+-})}
    \end{align}

    \noindent We cannot ignore the ordering in each line here as we are counting slightly different triangles $-+-$ versus $+--$, which depends on the edge ordering in the stream. However, by summing up the values and re-arranging, we get

    \begin{align}
        \Expect{R_{+--}} + \Expect{R_{-+-}} &= \sum_{(v < w) \in E}\sum_{u \in V}\tlk_{uvw,+--} + \tlk_{uwv,+--} + \sum_{(v < w) \in E}\sum_{u \in V}\tlk_{uvw,-+-} + \tlk_{uwv,-+-} \\
        &=\sum_{(u,v,w) \in V^3}\tlk_{uv,w;+--} + \sum_{(u,v,w) \in V^3}\tlk_{uv,w;-+-} \\
        &=\Tlk_{+--} + \Tlk_{-+-}.
    \end{align}

    In conclusion, $\Expect{\Rlk_1} = \Tlk_{+--} + \Tlk_{-+-} + \Tlk_{--+} = \Tlk_1$.
\end{proof}

Next, the upper bound on the variance of $\Rlk_1$ is straightforward.
\begin{proof}[(Proof of \cref{lem:X1_var})]
    This follows directly from $\abs{\Rlk_1} \le km$ and $\Expect{\Rlk_1} = \Tlk_1$.
\end{proof}

Lastly, we prove the desired performance and space bound on our median-of-means algorithm.
\begin{proof}[(Proof of \cref{lem:t1_less_k_restated})]
    Recall that $\Expect{\Rlk_1} = \Tlk_1$ and $\Var{\Rlk_1} \le (km)^2 - \left(\Tlk_1\right)^2$.

    Begin by averaging over $N_\epsilon$ independent copies of the quantum estimator, we get a value $\overline{X}_1$ with variance $Var(\overline{X}_1) = \frac{Var(\Rlk_1)}{N_\epsilon} \leq \frac{(km)^2 - (\Tlk_1)^2}{N_\epsilon}$. Using Chebyshev's inequality, we get

    \begin{equation}
        \prob{\abs{\overline{X}_1 - \Tlk_1} \geq \epsilon \Tlk_1} \leq \frac{(km)^2 - (\Tlk_1)^2}{N_\epsilon (\Tlk_1)^2}\frac{1}{\epsilon^2}
    \end{equation}
    
    \noindent Choosing $N_\epsilon \geq 4\frac{(km)^2 - (\Tlk_1)^2}{(\Tlk_1)^2}\frac{1}{\epsilon^2}$ results in $\prob{\abs{\overline{X}_1 - \Tlk_1} \leq \epsilon \abs{\Tlk_1}} \geq 3/4$.

    Repeat this process $N_\delta$ times to get estimates $\overline{X}_{1,1}, \dots, \overline{X}_{1,N_\delta}$ and define the variable

    \begin{equation}
        \hat{X}_1 := \mathtt{median}\left(\overline{X}_{1,1}, \dots, \overline{X}_{1,N_\delta} \right)
    \end{equation}
    
    \noindent We are interested in understanding when $\hat{X}_1$ is a good estimator for $\Rlk_1$. Define the Bernoulli indicator variable

    \begin{equation}
        Z_j := 1\left\{\abs{\overline{X}_{1,j} - \Tlk_1} \geq \epsilon \Tlk_1\right\}
    \end{equation}
    
    \noindent for when estimate $\overline{X}_{1,j}$ is bad. The median fails when over half of the estimators are bad. That is, when $Z \geq N_\delta / 2$ for $Z := \sum_j Z_j$. We know that $\Expect{Z} \leq N_\delta / 4$ by above. Applying Hoeffding's inequality we get

    \begin{equation}
        \prob{Z \geq N_\delta / 2} = \prob{Z - \Expect{Z} \geq N_\delta / 4} \leq \exp\left(-2N_\delta \left(\frac{1}{4}\right)^2\right) = \exp(-N_\delta/8)
    \end{equation}
    
    \noindent Using $N_\delta \geq 8 \ln \frac{1}{\delta}$ ensures that $\hat{X}_1$ is good with probability $1-\delta$.

    In total, we need to run $N_\epsilon \cdot N_\delta = 4\frac{(km)^2 - (\Tlk_1)^2}{(\Tlk_1)^2}\frac{1}{\epsilon^2} \cdot 8 \ln \frac{1}{\delta} = 32\frac{(km)^2 - (\Tlk_1)^2}{(\Tlk_1)^2}\frac{1}{\epsilon^2} \ln\frac{1}{\delta}$ iterations of the algorithm. Each iteration requires $2 \log (n) + 2$ qubits and $m = \mathcal{O}(\log n)$ classical bits to run. Therefore, the total space required for this overarching parallelized median-of-means algorithm is given by 

    \begin{equation}
        \mathcal{O} \left( \frac{(km)^2 - (\Tlk_1)^2}{(\Tlk_1)^2}  \log n\frac{1}{\epsilon^2}\log\frac{1}{\delta} \right).
    \end{equation}
\end{proof}

\end{document}